\documentclass[11 pt]{article}
\usepackage{fullpage}
\usepackage{xspace}
\usepackage{graphicx}
\usepackage{enumitem}
\usepackage{xcolor}

\usepackage{amsmath,amssymb,amsthm}
\newtheorem{theorem}{Theorem}[section]
\newtheorem{lemma}[theorem]{Lemma}
\newtheorem{corollary}[theorem]{Corollary}
\newtheorem{claim}{Claim}
\newtheorem{proposition}[theorem]{Proposition}
\newtheorem{definition}[theorem]{Definition}

\usepackage{algorithm,algorithmicx}
\usepackage{algpseudocode}

\newcommand{\algindent}{\hspace{\algorithmicindent}}

\newcommand{\kbscheme}{\ensuremath{\lambda_{\mathit{KB}}}}
\newcommand{\ackscheme}{\ensuremath{\lambda_{\mathit{ack}}}}
\newcommand{\gossipscheme}{\ensuremath{\lambda_{\mathit{gossip}}}}
\begin{document}
\def\thefootnote{\fnsymbol{footnote}}

\title{{Labeling Schemes for Deterministic Radio Multi-Broadcast}}

%
%
	\author{
	Colin Krisko\footnotemark[1]
	\and Avery Miller\footnotemark[2]
}
\footnotetext[1]{Department of Computer Science, University of Manitoba, Winnipeg, Manitoba, R3T 2N2, Canada.}
\footnotetext[2]{Department of Computer Science, University of Manitoba, Winnipeg, Manitoba, R3T 2N2, Canada. {\tt avery.miller@umanitoba.ca}. Supported by NSERC Discovery Grant RGPIN--2017--05936.}

\date{}
\maketitle              
\thispagestyle{empty}
\begin{abstract}
 We consider the multi-broadcast problem in arbitrary connected radio networks consisting of $n$ nodes. There are $k$ designated \emph{source nodes} for some fixed $k \in \{1,\ldots,n\}$, and each source node has a distinct piece of information that it wants to share with all nodes in the network. When $k=1$, this is known as the broadcasting problem, and when $k=n$, this is known as the gossiping problem. We consider the feasibility of solving multi-broadcast deterministically in radio networks. It is known that multi-broadcast is solvable when the nodes have distinct identifiers (e.g., using round-robin), and, it has been shown by Ellen, Gorain, Miller, and Pelc (2019) that the broadcasting problem is solvable if the nodes have been carefully assigned 2-bit labels rather than distinct identifiers. We set out to determine the shortest possible labels so that multi-broadcast can be solved deterministically in the labeled radio network by some universal deterministic distributed algorithm. 

First, we show that every radio network $G$ with maximum degree $\Delta$ can be labeled using $O(\min\{\log k,\log \Delta\})$-bit labels in such a way that multi-broadcast with $k$ sources can be accomplished. This bound is tight, in the sense that there are networks such that any labeling scheme sufficient for multi-broadcast with $k$ sources requires $\Omega(\min\{\log k,\log \Delta\})$-bit labels.

However, the above result is somewhat unsatisfactory: the bound is tight for certain network topologies (e.g., complete graphs), but there are networks where significantly shorter labels are sufficient. For example, we show how to construct a tree with maximum degree $\Theta(\sqrt{n})$ in which gossiping (i.e., multi-broadcast with $n$ sources) can actually be solved after labeling the nodes with $O(1)$-bit labels. So, we set out to find a labeling scheme for multi-broadcast that uses the optimal number of distinct labels in \emph{every} network. For all trees, we provide a labeling scheme and accompanying algorithm that will solve gossiping, and, we prove an impossibility result that demonstrates that our labeling scheme is optimal for gossiping in each tree. In particular, we prove that $\Theta(\log D(G))$-bit labels are necessary and sufficient in every tree $G$, where $D(G)$ denotes the \emph{distinguishing number of $G$}. This result also applies more generally to multi-broadcast in trees with $k \in \{2,\ldots,n\}$ sources in the case where the $k$ sources are not known when the labeling scheme is applied.
\end{abstract}
\section{Introduction}
Information dissemination is one of the fundamental goals for network algorithms. One important primitive is known as $k$-broadcast: in a network of $n$ nodes, there are $k$ \emph{source nodes} that each have some initial piece of information that they wish to share with all other nodes in the network.

We consider $k$-broadcast in synchronous radio networks, which is a particular model of wireless networks. More specifically, in a synchronous radio network, time proceeds in rounds, and each node in the network makes a decision in each round whether it will listen, or, transmit a message. In any round, a node receives a message if it listens and exactly one of its neighbours transmits. Otherwise, the node receives nothing, for one of three reasons: it is not listening, or, none of its neighbours are transmitting, or, two or more of its neighbours are transmitting (this case is known as a \emph{collision}, which models radio signal interference).

The possibility of collisions introduces an interesting challenge, as many simultaneous transmissions can prevent information from spreading in the network. In order to solve $k$-broadcast (as well as many other problems) there needs to be some way of breaking symmetry in the behaviour of the nodes, and one way this can be accomplished is by having each node use an assigned label or identifier during its execution. In fact, it is not hard to see that using assigned labels is necessary: it is impossible to deterministically solve 1-broadcast in a 4-cycle with unlabeled nodes (or identical labels) since, after the source node transmits its source message, its two neighbours will behave identically in all future rounds (i.e., both transmit or both listen), which means that the remaining node will never receive the source message. For this reason, many deterministic solutions to communication tasks in radio networks are designed for networks where each node has a unique identifier. In such networks, $k$-broadcast is always solvable using a simple round-robin algorithm: each node uses its unique identifier to ensure that it transmits in a round by itself, which avoids all transmission collisions, and this is repeated until all information has reached all nodes.

With regards to solving $k$-broadcast in radio networks using a deterministic algorithm, the previous paragraph presents two extremes: at least one-bit labels are required, since assigning the same label to all nodes leads to impossibility; however, $O(\log n)$-bit labels are sufficient, since assigning a unique binary string to each of the $n$ nodes and running a round-robin algorithm will eventually solve the task. This leads to a natural question: what is the shortest label size that allows us to solve $k$-broadcast in radio networks using a deterministic algorithm? A result by Ellen, Gorain, Miller, and Pelc \cite{EGMPSPAA} demonstrates that 2-bit labels are necessary and sufficient in the special case of 1-broadcast. In this work, we set out to answer this question more generally.

\subsection{The Model}
We consider networks modeled as simple undirected connected graphs with an arbitrary number of nodes $n$. For any fixed integer $k \in \{1,\ldots,n\}$, there are $k$ nodes $s_1,\ldots,s_k$ that are designated as \emph{sources}. For each $i \in \{1,\ldots,k\}$, source node $s_i$ initially has a \emph{source message} $\mu_i$.

Execution proceeds in synchronous rounds: each node has a local clock and all local clocks run at the same speed. We do not assume that the local time at each node is the same. Each node has a radio that it can use to send or receive transmissions. In each round, each node must choose one radio mode: \emph{transmit} or \emph{listen}. In transmit mode, a node sends an identical transmission to all its neighbours in the network. In listen mode, a node is silent and may receive transmissions. More specifically, in each round $t$ at each node $v$: 
\begin{itemize}
	\item If $v$ is in transmit mode in round $t$, then $v$ does not hear anything in round $t$.
	\item If $v$ is in listen mode in round $t$, and $v$ has no neighbours in transmit mode in round $t$, then $v$ does not hear anything in round $t$.
	\item If $v$ is in listen mode in round $t$, and $v$ has exactly one neighbour $w$ in transmit mode in round $t$, then $v$ receives the message contained in the transmission by $w$.
	\item If $v$ is in listen mode in round $t$, and $v$ has two or more neighbours in transmit mode in round $t$, then $v$ does not hear anything in round $t$. This case is often referred to as a \emph{collision}, and we assume that nodes have no way of detecting when a collision occurs.
\end{itemize}


\subsection{The Problem}
The \emph{$k$-broadcast problem} is solved when each node in the network possesses all of the source messages $\mu_1,\ldots,\mu_k$. Two well-known special cases of this problem are when $k=1$ (called \emph{broadcast}) and $k=n$ (called \emph{gossiping}). We will also consider a variant of this problem called \emph{acknowledged $k$-broadcast}, which requires that, at termination, all source nodes know that all nodes possess all of the source messages.

A \emph{labeling scheme} for a network $G = (V,E)$ is any function $\lambda$ from the set $V$ of nodes into the set of finite binary strings. This function $\lambda$ has complete information about $G$: the node set, the edge set, and the set of $k$ designated sources. The string $\lambda(v)$ is called the \emph{label} of the node $v$. Labels assigned by a labeling scheme are not necessarily distinct. The \emph{length} of a labeling scheme is the maximum label length taken over all network nodes. 


Suppose that each network $G$ has been labeled by some labeling scheme. We consider solving the $k$-broadcast task using a \emph{universal deterministic distributed algorithm}. In particular, each node initially knows its own label, and each of the $k$ source nodes $s_i$ possesses their source message $\mu_i$. In each round $t$, each node makes a decision whether it will transmit or listen in round $t$, and, if it decides to transmit, it decides on a finite binary string $m$ that it will send in its transmission during round $t$. These decisions are based \emph{only} on the current history of the node, that is: the label of the node, the node's source message (if it has one), and the sequence of messages received by the node before round $t$. 

Our goal in this paper is to determine the minimum possible length of a labeling scheme $\lambda$ such that there exists a universal deterministic distributed algorithm that solves $k$-broadcast in networks labeled with $\lambda$.

\subsection{Our Results}

In Section \ref{arbkbroad}, we show that every radio network $G$ with maximum degree $\Delta$ can be labeled using $O(\min\{\log k,\log \Delta\})$-bit labels in such a way that $k$-broadcast can be solved by a universal deterministic distributed algorithm, and we explicitly provide such an algorithm. In Section \ref{existential}, we demonstrate that this bound is tight, in the sense that there exist networks such that every labeling scheme sufficient for $k$-broadcast requires $\Omega(\min\{\log k,\log \Delta\})$-bit labels. However, in Section \ref{better}, we demonstrate that these bounds are not tight in every network: there exist trees on $n$ nodes with maximum degree $\Theta(\sqrt{n})$ in which $n$-broadcast (gossiping) can be solved after labeling each node with $O(1)$ bits. This inspires the question: can we prove tight bounds that hold in \emph{every} graph?

Restricting to the class of all trees, in Section \ref{trees}, we provide a labeling scheme and an accompanying universal deterministic distributed algorithm that solves gossiping, and we prove that the length of the labeling scheme is optimal for every graph in the class. In particular, we prove that $\Theta(\log D(G))$-bit labels are necessary and sufficient for solving gossiping in each tree $G$, where $D(G)$ denotes the distinguishing number of $G$, i.e., the smallest integer $c$ such that $G$ has a node labeling using $\{1,\ldots,c\}$ that is not preserved by any non-trivial graph automorphism. This result also applies more generally to $k$-broadcast in trees with $k \in \{2,\ldots,n\}$ sources in the case where the $k$ sources are not known when the labeling scheme is applied. From previous work about the distinguishing number of trees \cite{DistinguishTrees}, our bound can range anywhere from $\Theta(1)$ to $\Theta(\log n)$ depending on the tree, although it is known that the distinguishing number of any tree is bounded above by $\Delta$.

\subsection{Related Work}
Information dissemination tasks have been well-studied in radio network models in the context where each node has been pre-assigned a unique identifier. One set of results concerns centralized algorithms, i.e., each node has complete knowledge of the network. In this case, much is known about efficient deterministic solutions for broadcast \cite{AlonBarLinPel,Chlamtac,ChlamKut,ElkKort,GaberMans,GasPelXin,KowPelc}, gossiping \cite{GasPot,GasPotXin}, and multi-broadcast \cite{LevinKowSeg}. Another direction of research concerns distributed algorithms, i.e., each node initially only knows its own identifier. Again, there has been much progress in devising efficient  deterministic algorithms for broadcast \cite{ChlebGasGibPelcRyt,ChlebGasOstRob,ChrobGasRyt,ClemMontiSilvDistBcast}, gossiping \cite{ChristGasLing,ChrobGasRyt,GossipSurvey,GasLing,GasPagPotRad}, and multi-broadcast \cite{ChlebusMultiBcast,ClemMontiSilvMultiBcast,LevinKowSeg}.

As opposed to pre-assigned identifiers, there has been much previous work related to algorithmically solving tasks more efficiently after choosing labels for the nodes of the network, or choosing labels for mobile agents moving in the network (see related surveys \cite{OnlineAdviceSurvey,dobrev2013computing,IntroLocalCert,ilcinkasphd}). We restrict attention to previous work concerning tasks in radio networks. In \cite{GPtoptrees}, the authors proved that $\Theta(\log\log \Delta)$-bit labels are necessary and sufficient for topology recognition in radio networks with tree topology. When nodes have collision detectors, the authors of \cite{GPSizeAndDiam} proved that $\Theta(\log\log \Delta)$-bit labels are necessary and sufficient for computing the size of any radio network, while $O(1)$-bit labels are sufficient for computing the diameter of any radio network.  In \cite{fastradioadvice}, the authors considered the set of radio networks where broadcast is possible in $O(1)$ rounds when all nodes know the complete network topology, and they prove that broadcast is possible in such networks when the sum of the lengths of all labels is $O(n)$, and impossible when the sum is $o(n)$.

Most relevant to our work are previous results involving labeling schemes for information dissemination tasks. It was first shown in \cite{EGMP} that broadcast could be achieved in any radio network after applying a labeling scheme with length 2. The worst-case number of rounds used by their algorithm was $\Theta(n)$. In \cite{BuLevelSep}, the authors restricted attention to the class of level-separable radio networks, proved that 1-bit labels were sufficient in order to solve broadcast, and provided an accompanying broadcast algorithm using at most $2D$ rounds, where $D$ represents the source eccentricity. In \cite{EG}, the authors once again showed that $O(1)$-bit labels were sufficient for solving broadcast in any radio network, but they provided significantly faster broadcast algorithms: a non-constructive proof that an $O(D\log n + \log^2 n)$-round algorithm exists, and an explicit algorithm that completes within $O(D\log^2 n)$ rounds. In \cite{BuConverge}, the authors considered arbitrary radio networks, but instead studied the convergecast task: each node has an initial message, and all of these must eventually reach a designated sink node. They provide a labeling scheme using $O(\log n)$-bit labels, and an accompanying convergecast algorithm that uses $O(n)$ rounds. They prove that these bounds are tight for certain network topologies by proving matching lower bounds.

\section{Labeling Schemes and Algorithms for $k$-Broadcast in Arbitrary Graphs}\label{arbkbroad}

\subsection{Labeling Schemes and Algorithms for Acknowledged Broadcast}\label{EGMPAckBroadcast}
In this section, we recall some useful results from \cite{EGMP,EGMPSPAA} about solving acknowledged broadcast, and then extend them to define a new algorithm that will be used later in our work.

\subsubsection{Acknowledged Broadcast \cite{EGMPSPAA}} Given an arbitrary network $G$ with a designated start node $s_G$, there is a labeling scheme $\ackscheme$ that labels each node of $G$ using three bits called $join$, $stay$, and $ack$. There is a deterministic distributed algorithm $\mathcal{B}_{ack}$ that executes on the labeled version of $G$ that solves acknowledged broadcast: the designated start node $s_G$ possesses a message $\mu$, the message $\mu$ is eventually received by all other nodes in $G$, and after this occurs, the node $s_G$ eventually receives a message containing the string ``ack".  In fact, the algorithm $\mathcal{B}_{ack}$ can be viewed as two algorithms performed consecutively. First, a subroutine $\mathcal{B}$ is initiated by $s_G$, and this algorithm performs the broadcast of $\mu$ that eventually reaches all nodes, and, in the process, establishes a global clock (i.e., all nodes have their local clock value equal to $s_G$'s local clock at the end of the execution of $\mathcal{B}$). Then, a subroutine $\mathcal{ACK}$ is initiated by a node $z$ that sends the ``ack" message that is eventually received by $s_G$. The ``ack" message travels one hop per round along the same path that $\mu$ traveled from $s_G$ to $z$, but in reverse order, so the time to complete $\mathcal{ACK}$ is bounded above by the time to complete $\mathcal{B}$. In \cite{EGMPSPAA}, the unique node $z$ that initiates the $\mathcal{ACK}$ is designated by the labeling scheme using the label 001, and $z$ was chosen due to it being the last node to receive $\mu$ during the execution of $\mathcal{B}$. This choice of $z$ is important for the correctness of the algorithm, as it ensures that the execution of $\mathcal{B}$ is finished so that $\mathcal{ACK}$ can run on its own (which prevents transmissions from the two subroutines from interfering with one another). However, we note that any node could initiate $\mathcal{ACK}$, as long as it does so after the execution of $\mathcal{B}$ is finished. We will use this fact below to create a modified version of $\mathcal{B}_{ack}$ that will work in the case where we want to designate a specific node in $G$ as the initiator of the $\mathcal{ACK}$ subroutine. Another useful observation is that the labeling scheme $\ackscheme$ never sets the $join$, $stay$, and $ack$ bits all to 1 at any node. We will use this fact later to designate a special node in the network as a ``coordinator" by setting these three bits to 1, and this will not affect the original labeling or the behaviour of $\mathcal{B}_{ack}$.

\subsubsection{Bounded Acknowledged Broadcast \cite{EGMP}}
Using the same labeling scheme $\ackscheme$ as above, we describe a modification of $\mathcal{B}_{ack}$ so that it satisfies the following property: there is a common round $t_{done}$ in which all nodes know that the broadcast of the message $\mu$ has been completed, and, all nodes know an upper bound $m$ on how many rounds it took to complete the broadcast of $\mu$ (i.e., the number of rounds that elapse during the execution of $\mathcal{B}$). To implement this: when $z$ initiates the $\mathcal{ACK}$ algorithm, it takes note of the current round number $m$ and includes it in the ``ack" message, and when $s_G$ receives the ``ack" message, it initiates a broadcast using $\mathcal{B}$ with a new message $\mu'$ that contains the value of $m$. All nodes in the network will receive the broadcast of $\mu'$ before round $3m$: the broadcast of $\mu$ takes at most $m$ rounds, the $\mathcal{ACK}$ algorithm takes at most $m$ rounds, and the broadcast of $\mu'$ takes at most $m$ rounds. So the value $t_{done}=3m$ satisfies the desired property. We denote this version of the algorithm by $\mathcal{B}_{bounded}$.

\subsubsection{Acknowledged Broadcast with Designated Acknowledger}
Using the same labeling scheme $\ackscheme$ as above, and assuming that all nodes know an upper bound $m$ on the number of rounds that elapse during the execution of $\mathcal{B}$ (which could be learned by first executing $\mathcal{B}_{bounded}$, for example), we describe a modification of $\mathcal{B}_{ack}$ so that the acknowledgement process begins from a designated node $z_{des}$ that knows that it must initiate the $\mathcal{ACK}$ algorithm (e.g., it could be given a special label to indicate this). The algorithm consists of first executing $\mathcal{B}$, which performs the broadcast of $\mu$ starting at $s_G$, then, in round $m+1$, the node $z_{des}$ initiates the $\mathcal{ACK}$ algorithm. This algorithm works and completes within $2m$ rounds, since: the execution of $\mathcal{B}$ is finished by round $m$ (so all nodes know $\mu$ by round $m$), the execution of $\mathcal{ACK}$ takes at most an additional $m$ rounds, and, there is no interference between the $\mathcal{B}$ and $\mathcal{ACK}$ subroutines. We denote this version of the algorithm by $\mathcal{B}_{ack:des}$.


\subsection{Labeling Scheme $\kbscheme$ for $k$-Broadcast}\label{theKBscheme}

In this section, we provide a labeling scheme $\kbscheme$ that will be used by our algorithm $\mathcal{KB}$ (described in Section \ref{theKBalg}). At a high level, the algorithm $\mathcal{KB}$ works in two steps: first, the $k$ source messages are collected at a coordinator node (arbitrarily chosen by the labeling scheme), then, the coordinator broadcasts all the source messages to the entire network.

The number of bits used by our scheme is $O(\min\{\log k,\log \Delta\})$. To achieve this upper bound, two different labeling strategies are used, depending on the relationship between $k$ and $\Delta$. The labeling scheme has complete information about the network and the $k$ designated sources, which it uses to choose which labeling strategy to employ, and it uses a single bit in the labels to signal to the $k$-broadcast algorithm which strategy was used. 

At a high level, the labeling scheme follows one of the two following strategies:
\begin{enumerate}
	\item When the number of sources is no larger than the maximum degree of the graph (i.e., $k \leq \Delta$), the strategy is to give each source node a unique label so that, in the $k$-broadcast algorithm, each source node can send its source message to the coordinator, one at a time.
	\item When the number of sources is larger than the maximum degree of the graph (i.e., $k > \Delta$), the strategy is to label the network using a distance-two colouring so that, in the $k$-broadcast algorithm, a round-robin strategy that avoids transmission collisions can be used to share the source messages with the coordinator.
\end{enumerate}

To implement this idea, suppose that we are provided with a graph $G$ with $k$ designated source nodes $s_1,\ldots,s_k$. We assign a label to each node $v$ in $G$, and the label at each node $v$ consists of 5 components: a $strat$ bit, a $join$ bit, a $stay$ bit, an $ack$ bit, and a binary string $sched$. The labeling scheme assigns values to the components as follows:
\begin{enumerate}
	\item Choose an arbitrary node $r \in G$. This node will act as the \emph{coordinator}.
	\item Apply the labeling scheme $\ackscheme$ (see Section \ref{EGMPAckBroadcast}) to $G$ with designated start node $s_G=r$. This will set the $join$, $stay$, and $ack$ components (each consisting of one bit) at each node $v$. For the coordinator node $r$, set $join=stay=ack=1$.
	\item There are two cases:
	\begin{itemize}
		\item If $k \leq \Delta$, then: 
		\begin{enumerate}
			\item Set $strat$ to 0 at each node $v$.
			\item For each source node $s_i$ with $i \in \{1,\ldots,k\}$, set $sched$ at $s_i$ to be the binary representation of $i$.
			\item For each node $v \not\in \{s_1,\ldots,s_k\}$, set $sched$ at $v$ to the value 0.
		\end{enumerate}
		\item If $k > \Delta$, then: 
		\begin{enumerate}
			\item Set $strat$ to 1 at each node $v$.
			\item Compute a distance-two colouring of the graph $G$. More specifically, compute the graph $G^2$ and use the greedy colouring algorithm \cite{brooks} to properly colour the nodes of $G^2$ using integers, with smallest colour equal to 1. Let $c$ be the number of colours used.
			\item For each node $v$, set $sched$ at $v$ to be the $\lceil \log_2 (c+1) \rceil$-bit binary representation of the colour assigned to $v$ in the distance-two colouring of $G$.
		\end{enumerate}
	\end{itemize}

\end{enumerate}

%

\subsection{Algorithm $\mathcal{KB}$ for $k$-Broadcast}\label{theKBalg}
We now describe our deterministic distributed $k$-broadcast algorithm $\mathcal{KB}$ that is executed after the network nodes have been labeled using the labeling scheme $\kbscheme$ from Section \ref{theKBscheme}. The algorithm's execution consists of three subroutines performed consecutively: {\bf Initialize}, {\bf Aggregate}, and {\bf Inform}. These subroutines will make use of the algorithms $\mathcal{B}_{ack}$, $\mathcal{B}_{bounded}$, and $\mathcal{B}_{ack:des}$ described in Section \ref{EGMPAckBroadcast}.

The {\bf Initialize} subroutine consists of executing $\mathcal{B}_{bounded}$. The start node is the coordinator $r$ (the unique node with $join$, $stay$, and $ack$ bits all set to 1), and the broadcast message is ``init". At the conclusion of the execution, all nodes know an upper bound $m$ on the number of rounds that elapsed during the broadcast of the ``init" message, and, they all received this value before round $t_{done} = 3m$. Thus, in round $3m$, all nodes terminate the subroutine.

The {\bf Aggregate} subroutine is designed to collect all the source messages at the coordinator $r$. There are two possible algorithms, and the nodes will run one of these two algorithms depending on the value of the $strat$ bit that was set by the labeling scheme $\kbscheme$ (and this value is the same at all nodes). We present these two algorithms separately below.
\begin{itemize}
	\item If the $strat$ bit is 0, the nodes run \textsc{Individual-Collect}, which we now describe. The execution proceeds in phases, each consisting of exactly $2m$ rounds. For each $i \geq 1$, at the start of the $i^{th}$ phase, the coordinator $r$ initiates the $\mathcal{B}_{ack:des}$ algorithm with a broadcast message containing the value of $i$. In the $(m+1)^{th}$ round of the $i^{th}$ phase, the unique node that has its $sched$ bits set to the binary of representation of $i$ initiates the acknowledgement process, i.e., it will act as $z_{des}$. By the definition of $\kbscheme$, note that $z_{des} = s_i$. In its transmitted ``ack" message, $z_{des}$ includes its source message $\mu_i$. Eventually, there will be a phase $j$ in which the coordinator $r$ does not receive an ``ack" message, and, in phase $j+1$, the coordinator $r$ initiates the $\mathcal{B}_{ack}$ algorithm with broadcast message ``done". Upon receiving the ``done" message, each node terminates the \textsc{Individual-Collect} subroutine at the end of the current phase.
	\item If the $strat$ bit is 1, the nodes run \textsc{RoundRobin-Collect}, which we now describe. First, each node computes \texttt{numColours} using the calculation $2^{|sched|}-1$, where $|sched|$ represents the number of bits in the $sched$ part of its label (equivalently, the length of its label minus four). Then, the execution consists of $m$ phases, each consisting of exactly \texttt{numColours} rounds. In the $i^{th}$ round of each phase, a node transmits if and only if its $sched$ bits are equal to the binary representation of $i$, and its transmitted message is equal to the subset of source messages $\{\mu_1,\ldots,\mu_k\}$ that it knows. At the end of the $m^{th}$ phase, each node terminates the \textsc{RoundRobin-Collect} subroutine.
\end{itemize}

The {\bf Inform} subroutine consists of executing $\mathcal{B}_{bounded}$. The start node is the coordinator $r$, and the broadcast message is equal to the subset of source messages $\{\mu_1,\ldots,\mu_k\}$ that $r$ knows. All nodes terminate this subroutine at the same time, and they all know that $k$-broadcast has been completed.

\begin{theorem}\label{KBsolved}
	Consider any $n$-node unlabeled network $G$ with maximum degree $\Delta$, and, for any $k \in \{1,\ldots,n\}$, consider any designated source nodes $s_1,\ldots,s_k$ with source messages $\{\mu_1,\ldots,\mu_k\}$. By applying the labeling scheme $\kbscheme$ and then executing algorithm $\mathcal{KB}$, all nodes possess the complete set of source messages $\{\mu_1,\ldots,\mu_k\}$. The length of $\kbscheme$ is $O(\min\{\log k, \log \Delta\})$.
\end{theorem}
\begin{proof}
	First, we consider the length of the labeling scheme $\kbscheme$, as described in Section \ref{theKBscheme}. In the case where $k \leq \Delta$, the labeling scheme uses $O(\log k)$ bits: 1 bit for the $strat$ component, 3 bits for the $join$, $stay$, and $ack$ components, and $O(\log k)$ bits in the $sched$ component in order to store the binary representation of $i \in \{1,\ldots,k\}$ at each source node $s_i$ (and 1 bit at all non-source nodes). In the case where $k > \Delta$, the labeling scheme uses $O(\log \Delta)$ bits: 1 bit for the $strat$ component, 3 bits for the $join$, $stay$, and $ack$ components, and $O(\log \Delta)$ bits in the $sched$ component in order to represent the binary representation of the colour assigned by the distance-two colouring of $G$. To see why this last bound holds: recall by Brooks' Theorem \cite{brooks} that the greedy colouring algorithm applied to $G^2$ will use at most $1+\Delta_{G^2}$ colours (where $\Delta_{G^2}$ denotes the maximum degree of $G^2$), that $\Delta_{G^2} \in O(\Delta^2)$ (since, for each node in $G$, there are at most $\Delta^2$ other nodes in $G$ within distance 2), and that any integer bounded above by $\Delta^2$ has a binary representation using $O(\log \Delta^2) = O(\log \Delta)$ bits. Taken together, these two cases demonstrate that the length of the labeling scheme is $O(\min\{\log k, \log \Delta\})$.
	
	Next, we consider the correctness of $\mathcal{KB}$. Note that $\kbscheme$ assigns the same value to the $strat$ bit at all nodes, so all nodes are executing the same strategy. Also, regardless of strategy, the {\bf Initialize} and {\bf Inform} subroutines are identical. From the correctness of $\mathcal{B}_{bounded}$, at the end of {\bf Initialize}: all nodes know the same upper bound $m$ on the number of rounds that elapsed during the broadcast of the ``init" message, and all nodes terminated {\bf Initialize} in round $3m$. Also, from the correctness of $\mathcal{B}_{bounded}$, at the end of {\bf Inform}: all nodes know the subset of source messages that $r$ knows at the end of {\bf Aggregate}, and all nodes terminated {\bf Inform} at the same time. It remains to show that, regardless of the strategy used by {\bf Aggregate}, the coordinator $r$ possesses the entire set of source messages. In the case where $strat = 0$, the correctness of \textsc{Individual-Collect} follows almost immediately from the correctness of $\mathcal{B}_{ack:des}$: for each $i \in \{1,\ldots,k\}$, the $i^{th}$ phase ensures that $r$ receives the source message $\mu_i$ within $2m$ rounds. In the case where $strat=1$, the \textsc{RoundRobin-Collect} subroutine is a flooding algorithm that avoids all transmission collisions due to the distance-2 colouring of the nodes. It remains to confirm that it is executed for sufficiently many rounds. The length of each phase is $2^{|sched|}-1$, which is at least the number of colours $c$ used in the distance-2 colouring since $|sched| = \lceil \log_2 (c+1) \rceil$, so it follows that each node gets to transmit at least once per phase (without any collisions occurring at its neighbours). Finally, as $m$ is an upper bound on the distance between each node and the coordinator, $m$ phases are sufficient to ensure that each source message arrives at $r$ by the end of \textsc{RoundRobin-Collect}.
\end{proof}

\subsection{Existential Lower Bound}\label{existential}

In this section, we prove that in any complete graph $K_n$, any labeling scheme that is sufficient for solving $k$-broadcast has length at least $\Omega(\min\{\log k,\log \Delta\})$. This matches the $O(\min\{\log k,\log \Delta\})$ worst-case upper bound guaranteed by Theorem \ref{KBsolved}, which means that the upper bound cannot be improved in general. 

The idea behind the proof is to show that each source must be labeled differently by any labeling scheme: otherwise, using an indistinguishability argument, we prove that two sources with the same label will behave the same way in every round, which prevents any other node from receiving their source message (either due to both being silent, or both transmitting and causing collisions everywhere). Since any labeling scheme using at least $k$ different labels has length at least $\Omega(\log k)$, and $k \leq n = \Delta+1$, the result follows.

\begin{theorem}\label{existLower}
	Consider any integer $n > 1$, any $k \in \{1,\ldots,n\}$ and any labeling scheme $\lambda$. If there exists a universal deterministic distributed algorithm that solves $k$-broadcast on the complete graph $K_n$ labeled by $\lambda$, then the length of $\lambda$ is at least $\Omega(\min\{\log k,\log \Delta\})$.
\end{theorem}
\begin{proof}
	Consider any labeling scheme $\lambda$ and any universal deterministic distributed algorithm that solves $k$-broadcast on $K_n$ labeled by $\lambda$.
	
	First, we prove that the labels assigned to the $k$ sources by $\lambda$ are distinct. If $k=1$, there is nothing to prove, so we proceed with the cases where $k \geq 2$. To obtain a contradiction, suppose that two sources $s_\alpha$ and $s_\beta$ are assigned the same label $\ell$ by the labeling scheme $\lambda$.
	
	Consider the execution of $\mathcal{A}$ on the complete graph $K_n$ labeled by $\lambda$. We set out to prove that, in every round, $s_\alpha$ and $s_\beta$ perform the same action, i.e., either both listen, or both transmit. To do so, we consider node histories during the execution of $\mathcal{A}$: for an arbitrary node $x$, define $h_x[0]$ to be the label assigned to $x$ by $\lambda$, and, for each $t \geq 1$, define $h_x[t]$ to be the message received by $x$ in round $t$ (or $\bot$ if $x$ receives no message). As $\mathcal{A}$ is a deterministic algorithm, we know that $h_x[0\ldots t-1] = h_y[0\ldots t-1]$ implies that $x$ and $y$ perform the exact same action in round $t$, i.e., $x$ and $y$ both stay silent in round $t$, or, they both transmit the same message in round $t$. We apply this idea to the node histories $h_\alpha$ and $h_\beta$ of $s_\alpha$ and $s_\beta$, respectively, during the execution of $\mathcal{A}$.
	\begin{claim}
		For an arbitrary round $t \geq 1$, we have that $s_\alpha$ transmits in round $t$ if and only if $s_\beta$ transmits in round $t$.
	\end{claim}
	To prove the claim, it suffices to prove that, for each $t \geq 1$, nodes $s_\alpha$ and $s_\beta$ have the same history up to round $t-1$, i.e., $h_\alpha[0\ldots (t-1)] = h_\beta[0 \ldots (t-1)]$. This is because $\mathcal{A}$ is a deterministic distributed algorithm. We proceed by induction on the round number $t$. For the base case, consider $t=1$. By assumption, we have that $\lambda(s_\alpha) = \lambda(s_\beta) = \ell$, which implies that $h_\alpha[0] = \ell = h_{\beta}[0]$, as required. As induction hypothesis, assume that, for some $t \geq 2$, that $h_\alpha[0\ldots (t-2)] = h_{\beta}[0 \ldots (t-2)]$. For the inductive step, consider the possible cases for the value of $h_\alpha[t-1]$:
	\begin{itemize}
		\item {\bf Suppose that $h_\alpha[t-1] = \bot$.} There are several sub-cases to consider:
		\begin{itemize}
			\item {\bf $s_\alpha$ transmits in round $t-1$.} Then, by the induction hypothesis, we know that $h_\alpha[0\ldots (t-2)] = h_{\beta}[0 \ldots (t-2)]$, which implies that $s_\beta$ transmits in round $t-1$ as well. This means that $s_\beta$ does not receive a transmission in round $t-1$, so $h_{\beta}[t-1] = \bot$, as required.
			\item {\bf $s_\alpha$ does not transmit in round $t-1$, but two or more neighbours of $s_\alpha$ transmit in round $t-1$.} This means that two or more nodes transmit in round $t-1$. If $s_\beta$ is not one of these nodes, then it has them as neighbours in $K_n$, so a collision occurs at $s_\beta$ in round $t-1$, i.e., $h_{\beta}[t-1] = \bot$. Otherwise, if $s_\beta$ is one of the transmitting nodes, then $s_\beta$ does not receive a transmission in round $t-1$, so $h_{\beta}[t-1] = \bot$, as required.
			\item {\bf No node in $K_n$ transmits in round $t-1$.} Then $s_\beta$ has no transmitting neighbours in round $t-1$, so $s_\beta$ does not receive a transmission in round $t-1$, i.e., $h_{\beta}[t-1] = \bot$, as required.
		\end{itemize}
		\item {\bf Suppose that $h_\alpha[t-1] = m$ for some binary string $m$.} In particular, this means that $s_\alpha$ is in listen mode during round $t-1$, and has exactly one neighbour, say $s_\gamma$, that transmits in round $t-1$ (and this message is $m$). In $K_n$, this means that $s_\gamma$ is the only transmitting node in round $t-1$. By the induction hypothesis, we know that $h_\alpha[0\ldots (t-2)] = h_{\beta}[0 \ldots (t-2)]$, which implies that $s_\beta$ is in listen mode during round $t-1$ as well. It follows that $s_\beta \neq s_\gamma$, so, in $K_n$, $s_\gamma$ is a neighbour of $s_\beta$. As $s_\gamma$ is the only transmitting node, it follows that $s_\beta$ receives $m$ from $s_\gamma$ in round $t-1$. Thus, $h_\beta[t-1] = m$, as required.
	\end{itemize}
	In all cases, we showed that $h_\alpha[t-1] = h_{\beta}[t-1]$, which, together with the induction hypothesis, proves that $h_\alpha[0\ldots (t-1)] = h_{\beta}[0\ldots (t-1)]$. This completes the proof of the claim.
	
	Next, we use the above claim to reach the desired contradiction: that $\mathcal{A}$ does not solve $k$-broadcast. In particular, we prove that no node ever receives $\mu_\alpha$ (i.e., the source message that is initially only possessed by $s_\alpha$). Indeed, in each round $t$, either:
	\begin{itemize}
		\item $s_\alpha$ is in listen mode, in which case $\mu_\alpha$ is not transmitted by $s_\alpha$ in round $t$, or,
		\item $s_\alpha$ is in transmit mode, in which case $s_\beta$ is also in transmit mode (by the above claim), so a collision happens at all nodes in $K_n$. This implies that no node receives $\mu_\alpha$ in round $t$.
	\end{itemize}
	As this contradicts the correctness of $\mathcal{A}$, our assumption that $s_\alpha$ and $s_\beta$ are assigned the same label is incorrect, which concludes the proof that the labels assigned to the $k$ sources by $\lambda$ are distinct.
	
	Finally, in any set of at least $k$ distinct binary strings, there is at least one that has length $\Omega(\log k)$ bits, which proves that the labeling scheme $\lambda$ has length at least $\Omega(\log k)$. As $k \leq n = \Delta+1$, we conclude that $\Omega(\log k) \subseteq \Omega(\min\{\log k,\log \Delta\})$, as desired.
\end{proof}

\subsection{An Example of Better Labeling}\label{better}
In this section, for infinitely many values of $n$, we construct a tree $T_n$ on $n$ nodes with maximum degree $\Delta \in \Theta(\sqrt{n})$ such that, after labeling each node with $O(1)$ bits, the $n$-broadcast task (gossiping) can be solved by a deterministic distributed algorithm. The length of the labeling scheme is significantly than the upper and lower bounds from Theorems \ref{KBsolved} and \ref{existLower}, which for $T_n$ would give $\Theta(\min\{\log k,\log \Delta\}) = \Theta(\log n)$.

Let $n$ be any triangular number greater than 1, i.e., there exists a positive integer $x \geq 2$ such that $n = x(x+1)/2$. The tree $T_n$ consists of: a node $r$, and, for each $i \in \{2,\ldots,x\}$, a node $\ell_i$ and a path of length $i$ with endpoints $\ell_i$ and $r$.  Note that node $r$ has degree $x-1 \in \Theta(\sqrt{n})$, each node in $\{\ell_2,\ldots,\ell_x\}$ has degree 1, and all other nodes have degree 2.

We label each node of $T_n$ with 2 bits: node $r$ is given the label 11, each node in $\{\ell_2,\ldots,\ell_{x-1}\}$ is given the label 01, the node $\ell_x$ is given the label 10, and all other nodes are given the label 00.

To solve $n$-broadcast, the idea is to initiate a Broadcast from node $r$ to send an ``init" message that gets forwarded along the paths towards each leaf $\ell_2,\ldots,\ell_x$, and, when each leaf receives the ``init" message, it sends a ``gather" message back along the path towards $r$. Each time a ``gather" message is forwarded by a node $v$, the node appends its own source message. The fact that the paths have distinct lengths means that the ``gather" messages along each path arrive back at $r$ at different times, which prevents collisions at $r$. When $\ell_x$ sends its ``gather" message, it also includes the string ``last", and when this message is received by $r$, it initiates another Broadcast of a ``spread" message containing all of the source messages. Of all the leaf-to-$r$ paths, the one involving leaf $\ell_x$ is the longest, which means that the ``gather" message containing ``last" is the last one that $r$ receives, and this guarantees that $r$ possesses all of the source messages before initiating the final Broadcast. Algorithm \ref{TnGossipPseudo} gives the pseudocode of the $n$-broadcast algorithm executing at each node in the labeled $T_n$.

We now prove that the algorithm solves $n$-broadcast. For each $i \in \{2,\ldots,x\}$, denote by $P_i$ the path of length $i$ with endpoints $r$ and $\ell_i$. In what follows, all round numbers are referring to $r$'s local clock (which starts at round $t=1$).
\begin{lemma}\label{lem:initbcast}
	For each $i \in \{2,\ldots,x\}$ and each $j \in \{1,\ldots,i\}$, the node at distance $j$ from $r$ in path $P_i$ receives an $\langle ``init" \rangle$ message for the first time in round $j$, and no other node in $P_i - \{r\}$ receives a message containing $``init"$ or $``gather"$ or $``spread"$ for the first time in round $j$.
\end{lemma}
\begin{proof}
	Consider an arbitrary path $P_i$ for $i \in \{2,\ldots,x\}$. In what follows, we will denote by $v_d$ the node at distance $d$ from $r$ on path $P_i$. We prove the result by induction on $j$.
	
	For the base case, consider $j=1$. Note that node $v_1$ is the only neighbour of $r$ in $P_i$. According to lines \ref{t1cond}-\ref{t1transmit} of the algorithm, the node with label 11 transmits $\langle ``init" \rangle$ in round 1 (and not before). According to our labeling, only the node $r$ has label 11. By the definition of the algorithm, we see that all other transmissions occur in a round immediately following a round in which a message was received for the first time. It follows that $r$ is the only node that transmits in round 1, and no node transmits in an earlier round. It follows that $v_1$ receives an $\langle ``init" \rangle$ message for the first time in round 1, and is the only node in $P_i$ that receives any message.
	
	As induction hypothesis, assume that for some $j \in \{1,\ldots,i-1\}$, for each $j' \in \{1,\ldots,j\}$, node $v_{j'}$ receives an $\langle ``init" \rangle$ message for the first time in round $j'$, and no other node in $P_i - \{r\}$ receives a message containing $``init"$ or $``gather"$ or $``spread"$ for the first time in round $j'$.
	
	Consider the node $v_{j+1}$. By the induction hypothesis, in round $j$, node $v_j$ receives an $\langle ``init" \rangle$ message for the first time. Thus, in round $j+1$, node $v_j$ transmits. Further, by the induction hypothesis, no other node in $P_i - \{r\}$ receives a message containing $``init"$ or $``gather"$ or $``spread"$ for the first time in round $j$, so, by the definition of the algorithm, no other node in $P_i - \{r\}$ transmits in round $j+1$. Therefore, the only transmission that occurs in round $j+1$ by nodes in $P_i - \{r\}$ is the transmission of $\langle ``init" \rangle$ by $v_j$. In particular, it follows that no collision occurs at $v_{j-1}$ and $v_{j+1}$, so they both receive an $\langle ``init" \rangle$ message in round $j+1$. However, by the induction hypothesis, we know that $v_{j-1}$ received an $\langle ``init" \rangle$ message for the first time in round $j-1$, so round $j+1$ is not the first time it receives an $\langle ``init" \rangle$ message. Further, by the induction hypothesis, node $v_{j+1}$ does not receive a message containing $``init"$ or $``gather"$ or $``spread"$ for the first time in any of the rounds $1,\ldots,j$. It follows that $v_{j+1}$ is the only node in $P_i$ that receives an $\langle ``init" \rangle$ message for the first time in round $j+1$, and, no other nodes receive a message containing $``init"$ or $``gather"$ or $``spread"$ for the first time in round $j+1$, as desired.
\end{proof}
For an arbitrary $i \in \{2,\ldots,x\}$, node $\ell_i$ is the unique node in $P_i$ that is distance $i$ from $r$, so the following result follows from Lemma \ref{lem:initbcast} with $j=i$.
\begin{corollary}
	For each $i \in \{2,\ldots,x\}$, node $\ell_i$ receives an $\langle ``init" \rangle$ message in round $i$.
\end{corollary}

\begin{lemma}\label{lem:gathercast}
	For each $i \in \{2,\ldots,x\}$ and each $j \in \{1,\ldots,i-1\}$, the node at distance $j$ from $\ell_i$ in path $P_i$ receives a message in round $i+j$ containing $``gather"$ for the first time, the message also contains all source messages of nodes within distance $j-1$ of $\ell_i$, and no other node in $P_i - \{r\}$ receives a message containing $``init"$ or $``gather"$ or $``spread"$ for the first time in round $i+j$.
\end{lemma}
\begin{proof}
	Consider an arbitrary path $P_i$ for $i \in \{2,\ldots,x\}$. In what follows, we will denote by $v_d$ the node at distance $d$ from $\ell_i$ on path $P_i$. We prove the result by induction on $j$.
	
	For the base case, consider $j=1$. Note that node $v_1$ is the only neighbour of $\ell_i$ in $P_i$. By Lemma \ref{lem:initbcast}, node $\ell_i$ receives an $\langle ``init" \rangle$ message for the first time in round $i$. Recall that, by our labeling, node $\ell_i$ is either labeled 01 (if $i \in \{2,\ldots,x-1\}$, or labeled 10 (if $i = x$). In both cases, after receiving an $\langle ``init" \rangle$ message in round $i$, node $\ell_i$ transmits a message containing $``gather"$ in round $i+1$, and also its value of \texttt{knownMsgs}, which contains $\ell_i$'s source message. Also, by Lemma \ref{lem:initbcast}, no other node in $P_i-\{r\}$ receives a message containing $``init"$ or $``gather"$ or $``spread"$ for the first time in round $i$, so, by the definition of the algorithm, no other node in $P_i-\{r\}$ transmits in round $i+1$. It follows that no collision occurs at node $v_1$ in round $i+1$, so it receives the $``gather"$ message in round $i+1$ containing $\ell_i$'s source message, and no other node in $P_i - \{r\}$ receives a message containing $``init"$ or $``gather"$ or $``spread"$ for the first time in round $i+1$.
	
	As induction hypothesis, assume that for some $j \in \{1,\ldots,i-2\}$, for each $j' \in \{1,\ldots,j\}$, node $v_{j'}$ receives a message in round $i+j'$ containing $``gather"$ for the first time, the message also contains all source messages of nodes within distance $j'-1$ of $\ell_i$, and no other node in $P_i - \{r\}$ receives a message containing $``init"$ or $``gather"$ or $``spread"$ for the first time in round $i+j'$.
	
	Consider the node $v_{j+1}$. By the induction hypothesis, in round $i+j$, node $v_j$ receives a message containing $``gather"$ for the first time, and the message also contains all source messages of nodes within distance $j'-1$ of $\ell_i$. By the definition of the algorithm, in round $i+j+1$, node $v_j$ transmits a message containing all of the messages it knows (including its own), so its message contains all source messages of nodes within distance $j$ of $\ell_i$.
	Further, by the induction hypothesis, no other node in $P_i - \{r\}$ receives a message containing $``init"$ or $``gather"$ or $``spread"$ for the first time in round $i+j$, so, by the definition of the algorithm, no other node in $P_i - \{r\}$ transmits in round $i+j+1$. Therefore, the only transmission that occurs in round $i+j+1$ by nodes in $P_i - \{r\}$ is the message transmitted by $v_j$. In particular, it follows that no collision occurs at $v_{j-1}$ and $v_{j+1}$, so they both receive $v_j$'s $``gather"$ message in round $i+j+1$ (and this message contains all source messages of nodes within distance $j$ of $\ell_i$). However, by the induction hypothesis, we know that $v_{j-1}$ received a message containing $``gather"$ for the first time in round $i+j-1$, so round $i+j-1$ is not the first time it receives a message containing $``gather"$. Further, by the induction hypothesis, node $v_{j+1}$ does not receive a message containing $``init"$ or $``gather"$ or $``spread"$ for the first time in any of the rounds $i+1,\ldots,i+j$. It follows that $v_{j+1}$ is the only node in $P_i - \{r\}$ that receives a $``gather"$ message for the first time in round $j+1$, and, no other nodes in $P_i - \{r\}$ receive a message containing $``init"$ or $``gather"$ or $``spread"$ for the first time in round $i+j+1$, as desired.
\end{proof}

By Lemma \ref{lem:gathercast}, for the neighbour $w$ of $r$ on an arbitrary path $P_i$, node $w$ receives a message containing $``gather"$ for the first time in round $2i-1$, and this message contains all source messages of nodes in $P_i - \{r,w\}$. By the definition of the algorithm, in round $2i$, node $w$ transmits a message containing $``gather"$ as well as all of the source messages it knows (including its own), and is the only neighbour to do so since any other neighbour of $r$ is on a path $P_{i'}$ with $i' \neq i$. Namely, for each $i \in \{2,\ldots,x\}$, node $r$ receives a message in round $2i$ containing the source messages of all nodes on path $P_i - \{r\}$. When $i=x$, the message sent by $w$ also contains $``last"$. Hence, we get the following two results.
\begin{corollary}\label{cor:rcvgather}
	For each $i \in \{2,\ldots,x-1\}$, node $r$ receives a message in round $2i$ containing ``gather", and the message contains all source messages of nodes in $P_i - \{r\}$.
\end{corollary}
\begin{corollary}\label{cor:rcvlast}
	Node $r$ receives a message in round $2x$ containing both $``gather"$ and $``last"$ for the first time, and the message contains all source messages of nodes in $P_x - \{r\}$.
\end{corollary}
Finally, after round $2x$, node $r$ possesses all of the source messages, and initiates a final Broadcast of a message containing all of them. The following result shows that this message is eventually received by all nodes, which immediately implies that the algorithm is correct.

\begin{lemma}
	For each $i \in \{2,\ldots,x\}$ and each $j \in \{1,\ldots,i\}$, the node at distance $j$ from $r$ in path $P_i$ receives a message containing $``spread"$ for the first time in round $2x+j$, the message contains all the source messages, and no other node in $P_i - \{r\}$ receives a message containing $``init"$ or $``gather"$ or $``spread"$ for the first time in round $2x+j$.
\end{lemma}
\begin{proof}
		Consider an arbitrary path $P_i$ for $i \in \{2,\ldots,x\}$. In what follows, we will denote by $v_d$ the node at distance $d$ from $r$ on path $P_i$. We prove the result by induction on $j$.
	
	For the base case, consider $j=1$. Note that node $v_1$ is the only neighbour of $r$ in $P_i$. By Corollary \ref{cor:rcvlast}, node $r$ receives a message in round $2x$ containing both $``gather"$ and $``last"$ for the first time. By our labeling scheme, node $r$ has label 11, so, by the definition of the algorithm, node $r$ transmits a message in round $2x+1$ containing $``spread"$ and all of the source messages it possesses. Further, it follows from Corollary \ref{cor:rcvgather} that, by round $2x$, node $r$ has received the source messages of all nodes in $T_n$, so $r$'s transmitted message in round $2x+1$ contains all of the source messages. By Lemma \ref{lem:gathercast}, no other node in $P_i-\{r\}$ receives a message containing $``init"$ or $``gather"$ or $``spread"$ for the first time in round $2x$, so no node in $P_i-\{r\}$ transmits in round $2x+1$. Therefore, no collision occurs at $v_1$ in round $2x+1$, so $v_1$ receives $r$'s transmitted message containing $``spread"$ and all of the source messages, and, moreover, no node in $P_i - \{r,v_1\}$ receives a message containing $``init"$ or $``gather"$ or $``spread"$ for the first time in round $2x+1$.
	
	As induction hypothesis, assume that for some $j \in \{1,\ldots,i-1\}$, for each $j' \in \{1,\ldots,j\}$, node $v_{j'}$ receives a message containing $``spread"$ for the first time in round $2x+j'$, the message contains all the source messages, and no other node in $P_i - \{r\}$ receives a message containing $``init"$ or $``gather"$ or $``spread"$ for the first time in round $2x+j'$.
	
	Consider the node $v_{j+1}$. By the induction hypothesis, in round $2x+j$, node $v_j$ receives a message containing $``spread"$ for the first time. Thus, in round $2x+j+1$, node $v_j$ transmits. Further, by the induction hypothesis, no other node in $P_i - \{r\}$ receives a message containing $``init"$ or $``gather"$ or $``spread"$ for the first time in round $2x+j$, so, by the definition of the algorithm, no other node in $P_i - \{r\}$ transmits in round $2x+j+1$. Therefore, the only transmission that occurs in round $2x+j+1$ by nodes in $P_i - \{r\}$ is the transmission of a message containing ``spread" by $v_j$, and this transmission contains all the source messages. In particular, it follows that no collision occurs at $v_{j-1}$ and $v_{j+1}$, so they both receive $v_j$'s message in round $2x+j+1$. However, by the induction hypothesis, we know that $v_{j-1}$ received a message containing $``spread"$ for the first time in round $2x+j-1$, so round $2x+j+1$ is not the first time it receives an message containing $``spread"$. Further, by the induction hypothesis, node $v_{j+1}$ does not receive a message containing $``init"$ or $``gather"$ or $``spread"$ for the first time in any of the rounds $2x+1,\ldots,2x+j$. It follows that $v_{j+1}$ is the only node in $P_i - \{r\}$ that receives a message containing $``spread"$ for the first time in round $2x+j+1$, and, no other nodes in $P_i-\{r\}$ receive a message containing $``init"$ or $``gather"$ or $``spread"$ for the first time in round $2x+j+1$, as desired.
\end{proof}

\begin{algorithm}[H]
	\small
	\caption{The $n$-broadcast algorithm at each node $v$ in labeled $T_n$}
	\label{TnGossipPseudo}
	\begin{algorithmic}[1]
		\Statex {\color{gray} \% Node $v$ has a source message $\mu_v$, the node's two-bit label is stored in \texttt{label}}
		\State $\texttt{knownMsgs} \leftarrow \{\mu_v\}$
		\For{each round $t \geq 1$}
		\If{$\texttt{label} == 11$}{\color{gray}\algindent \% this is node $r$}
			\If{$t == 1$}\label{t1cond}
				\State transmit $\langle ``init" \rangle$\label{t1transmit}
			\ElsIf{received $\langle \textrm{``gather"},msgs \rangle$ in round $t-1$}
				\State $\texttt{knownMsgs} \leftarrow \texttt{knownMsgs} \cup msgs$
			\ElsIf{received $\langle \textrm{``gather"},\textrm{``last"},msgs \rangle$ in round $t\!-\!1$ for the first time}
				\State $\texttt{knownMsgs} \leftarrow \texttt{knownMsgs} \cup msgs$
				\State transmit $\langle \textrm{``spread"},\texttt{knownMsgs} \rangle$
			\Else
				\State listen for a message
			\EndIf
		\EndIf
		\If{$\texttt{label} == 01$}{\color{gray}\algindent \% this is one of the nodes $\ell_2,\ldots,\ell_{x-1}$}
			\If{received $\langle \textrm{``init"} \rangle$ in round $t-1$ for the first time}
				\State transmit $\langle \textrm{``gather"},\texttt{knownMsgs} \rangle$
			\ElsIf{received $\langle \textrm{``spread"},msgs\rangle$ in round $t-1$ for the first time}
				\State $\texttt{knownMsgs} \leftarrow \texttt{knownMsgs} \cup msgs$
				\State terminate()
			\Else
				\State listen for a message
			\EndIf
		\EndIf
		\If{$\texttt{label} == 10$}{\color{gray}\algindent \% this is node $\ell_{x}$}
			\If{received $\langle \textrm{``init"} \rangle$ in round $t-1$ for the first time}
			\State transmit $\langle \textrm{``gather"},\textrm{``last"},\texttt{knownMsgs} \rangle$
			\ElsIf{received $\langle \textrm{``spread"},msgs \rangle$ in round $t-1$ for the first time}
			\State $\texttt{knownMsgs} \leftarrow \texttt{knownMsgs} \cup msgs$
			\State terminate()
			\Else
			\State listen for a message
			\EndIf
		\EndIf
		\If{$\texttt{label} == 00$}{\color{gray}\algindent \% this is an internal node on an $r$-to-$\ell_{i}$ path}
			\If{received $\langle \textrm{``init"} \rangle$ in round $t-1$ for the first time}
				\State transmit $\langle \textrm{``init"}\rangle$
			\ElsIf{received $\langle \textrm{``gather"},msgs \rangle$ in round $t-1$ for the first time}
				\State $\texttt{knownMsgs} \leftarrow \texttt{knownMsgs} \cup msgs$
				\State transmit $\langle \textrm{``gather"},\texttt{knownMsgs} \rangle$
			\ElsIf{received $\langle \textrm{``gather"},\textrm{``last"},msgs \rangle$ in round $t\!-\!1$ for the first time}
				\State $\texttt{knownMsgs} \leftarrow \texttt{knownMsgs} \cup msgs$
				\State transmit $\langle \textrm{``gather"},\textrm{``last"},\texttt{knownMsgs} \rangle$
			\ElsIf{received $\langle \textrm{``spread"},msgs \rangle$ in round $t-1$ for the first time}
				\State $\texttt{knownMsgs} \leftarrow \texttt{knownMsgs} \cup msgs$
				\State transmit $\langle \textrm{``spread"},\texttt{knownMsgs} \rangle$
				\State terminate()
			\Else
				\State listen for a message
			\EndIf
		\EndIf
	
		\EndFor
	\end{algorithmic}
\end{algorithm}

\begin{corollary}
	There exists a $O(1)$-bit labeling scheme and a deterministic distributed algorithm that will solve $k$-broadcast in $T_n$, for all $k \in \{1,\ldots,n\}$.
\end{corollary}

\section{Labeling Schemes and Algorithms for Gossiping in Trees}\label{trees}
In this section, we restrict to the class of trees and give an optimal labeling scheme that allows us to solve gossiping. (This also applies to $k$-broadcast for each $k \in \{1,\ldots,n-1\}$, as we could use the same algorithm and have the non-sources operate with a blank source message.) The labeling scheme has length $O(\log D(G))$ bits for every tree $G$, where $D(G)$ is the distinguishing number of $G$. First, we present the lower bound that demonstrates that this quantity is optimal: any labeling scheme that is sufficient for gossiping to be solved in a graph $G$ must use labels of size $\Omega(\log D(G))$. In fact, the proof of this lower bound also works in the case of $k$-broadcast for each $k \in \{2,\ldots,n-1\}$, however, only under the additional condition that the source nodes are not known when the labeling scheme is applied.

Before proceeding, we review some relevant graph-theoretic concepts that will be used to derive both the upper and lower bounds. Definition \ref{distdefs} reviews some terminology regarding distances, Definition \ref{autodefs} reviews definitions relating to graph automorphisms, and Definition \ref{distinguishdefs} defines the distinguishing number.

\begin{definition}\label{distdefs}
	Let $G = (V,E)$ be a graph. For every $u,v \in V$, the \emph{distance between $u$ and $v$}, denoted by $d(u,v)$, is the length of a shortest path between $u$ and $v$ in $G$. For any $v \in V$, the \emph{eccentricity of $v$}, denoted by $\mathrm{ecc}(v)$, is the largest distance from $v$ to any other node in $G$, i.e., $\max_{u \in V}\{d(u,v)\}$. The \emph{radius of $G$}, denoted by $\mathrm{rad}(G)$, is the smallest node eccentricity in $G$, i.e., $\min_{v \in V}\{ecc(v)\}$. The set of nodes with minimum eccentricity, i.e., $\{v \in V\ |\ \mathrm{ecc}(v) = \mathrm{rad}(G)\}$, is called the \emph{center of $G$}, and is denoted by $\mathrm{center}(G)$. For any node $v \in V$, we define the \emph{distance between $v$ and $\mathrm{center}(G)$} to be $\min_{u \in \mathrm{center}(G)} d(u,v)$.
\end{definition}

\begin{definition}\label{autodefs}
	Let $G = (V,E)$ be a graph. A bijection $\phi : V \rightarrow V$ is called an \emph{automorphism of $G$} if, for every $u,v \in V$, we have that $\{u,v\} \in E$ if and only if $\{\phi(u),\phi(v)\} \in E$. An automorphism $\phi$ is \emph{non-trivial} if there exists a $v \in V$ such that $\phi(v) \neq v$. For a fixed graph $G$, the set of all automorphisms of $G$ is denoted by $\mathrm{Aut}(G)$.
\end{definition}

\begin{definition}[Albertson and Collins \cite{albertsoncollins}]\label{distinguishdefs} 
	A labeling $\rho : V \rightarrow \{1,\ldots,c\}$ is called \emph{$c$-distinguishing} if, for every non-trivial $\phi \in \mathrm{Aut}(G)$, there exists $v \in V$ such that $\rho(v) \neq \rho(\phi(v))$. The \emph{distinguishing number of $G$}, denoted by $D(G)$, is the smallest integer $c$ such that $G$ has a labeling that is $c$-distinguishing.
\end{definition}

\subsection{Lower Bound}
First, we prove the following three facts that demonstrate various properties preserved by graph automorphisms: distance between nodes, $\mathrm{center}(G)$, and distance to $\mathrm{center}(G)$.

\begin{proposition}\label{preservedist}
	For any graph $G = (V,E)$, any $x,y \in V$, and any $\phi \in \mathrm{Aut}(G)$, we have $d(x,y) = d(\phi(x),\phi(y))$.
\end{proposition}
\begin{proof}
	Let $d(x,y) = m$, and let $d(\phi(x),\phi(y)) = m'$.
	
	By definition of distance, $d(x,y) = m$ means that there is a path $(v_1,\ldots,v_{m+1})$ where $x=v_1$, $y=v_{m+1}$, and $\{v_i,v_{i+1}\} \in E$ for each $i \in \{1,\ldots,m\}$. By the definition of automorphism, $\phi$ preserves vertex adjacency, i.e., $\{u,v\} \in E \Leftrightarrow \{\phi(u),\phi(v)\} \in E$, so there is a path $(\phi(v_1),\ldots,\phi(v_{m+1}))$ where $x=\phi(v_1)$, $y=\phi(v_{m+1})$, and $\{\phi(v_i),\phi(v_{i+1})\} \in E$ for each $i \in \{1,\ldots,m\}$. This provides a path between $\phi(x)$ and $\phi(y)$ with length $m$, which proves that $d(\phi(x),\phi(y)) = m' \geq m$. 
	
	 By definition of distance, $d(\phi(x),\phi(y)) = m'$ means that there is a path $(w_1,\ldots,w_{m'+1})$ where $\phi(x)=w_1$, $\phi(y)=w_{m+1}$, and $\{w_i,w_{i+1}\} \in E$ for each $i \in \{1,\ldots,m\}$. As $\phi$ is a bijection, we can re-write each $w_i = \phi(\phi^{-1}(w_i))$, and represent the path between $\phi(x)$ and $\phi(y)$ as $(\phi(\phi^{-1}(w_1)),\ldots,\phi(\phi^{-1}(w_{m'+1})))$. By the definition of automorphism, $\phi$ preserves vertex adjacency, i.e., $\{u,v\} \in E \Leftrightarrow \{\phi(u),\phi(v)\} \in E$, so there is a path $(\phi^{-1}(w_1),\ldots,\phi^{-1}(w_{m'+1}))$ where $x=\phi^{-1}(w_1)$, $y=\phi^{-1}(w_{m+1})$, and $\{\phi^{-1}(w_i),\phi^{-1}(w_{i+1})\} \in E$ for each $i \in \{1,\ldots,m\}$. This provides a path between $x$ and $y$ with length $m'$, which proves that $d(x,y) = m \geq m'$. 
	
	Altogether, we have shown that $m = m'$, as desired.
\end{proof}

\begin{proposition}\label{preservecenter}
	For any graph $G = (V,E)$ and any $\phi \in \mathrm{Aut}(G)$, we have $v \in \mathrm{center}(G)$ if and only if $\phi(v) \in \mathrm{center}(G)$.
\end{proposition}
\begin{proof}
	Suppose that $v \in \mathrm{center}(G)$. By definition, $v \in \mathrm{center}(G)$ implies that there exists a node $u$ such that $d(u,v) = \mathrm{rad}(G)$. By Proposition \ref{preservedist}, we know $d(u,v) = d(\phi(u),\phi(v))$, i.e., there exists a node (namely, $\phi(u)$) at distance $\mathrm{rad}(G)$ from $\phi(v)$. It follows that $\mathrm{ecc}(\phi(v)) \geq \mathrm{rad}(G)$. By the definition of $\mathrm{rad}(G)$, all nodes have eccentricity at most $\mathrm{rad}(G)$, so we conclude that $\mathrm{ecc}(\phi(v)) = \mathrm{rad}(G)$. This means that $\phi(v) \in \mathrm{center}(G)$.
	
	Conversely, suppose that $\phi(v) \in \mathrm{center}(G)$. By definition, $\phi(v) \in \mathrm{center}(G)$ implies that there exists a node $w$ such that $d(w,\phi(v)) = \mathrm{rad}(G)$. As $\phi$ is a bijection, we can write $\phi(\phi^{-1}(w))$, so $d(\phi(\phi^{-1}(w)),\phi(v)) = \mathrm{rad}(G)$. By Proposition \ref{preservedist}, we know $d(\phi(\phi^{-1}(w)),\phi(v)) = d(\phi^{-1}(w),v)$, i.e., there exists a node (namely, $\phi^{-1}(w)$) at distance $\mathrm{rad}(G)$ from $v$. It follows that $\mathrm{ecc}(v) \geq \mathrm{rad}(G)$. By the definition of $\mathrm{rad}(G)$, all nodes have eccentricity at most $\mathrm{rad}(G)$, so we conclude that $\mathrm{ecc}(v) = \mathrm{rad}(G)$. This means that $v \in \mathrm{center}(G)$.
\end{proof}
\begin{corollary}\label{preservedistcenter}
	For any graph $G = (V,E)$, any $\phi \in \mathrm{Aut}(G)$, and any $v \in V$, the distance between $v$ and $\mathrm{center}(G)$ is equal to the distance between $\phi(v)$ and $\mathrm{center}(G)$.
\end{corollary}
\begin{proof}
	Let $m$ be the distance between $v$ and $\mathrm{center}(G)$. Let $m'$ be the distance between $\phi(v)$ and $\mathrm{center}(G)$.
	
	Let $u$ be any node in $\mathrm{center}(G)$ such that $d(u,v) = m$. By Proposition \ref{preservedist}, we get that $d(\phi(u),\phi(v)) = d(u,v) = m$, and, by Proposition \ref{preservecenter}, node $\phi(u) \in \mathrm{center}(G)$. It follows that the distance between $\phi(v)$ and $\mathrm{center}(G)$ is at most $m$, i.e., $m' \leq m$.
	
	 Let $w$ be any node in $\mathrm{center}(G)$ such that $d(w,\phi(v)) = m'$. Since $\phi$ is a bijection, we can write $w = \phi(\phi^{-1}(w))$, so $d(\phi(\phi^{-1}(w)),\phi(v)) = m'$. By Proposition \ref{preservedist}, we get that $d(\phi^{-1}(w),v) = d(\phi(\phi^{-1}(w)),\phi(v)) = m'$, and, by Proposition \ref{preservecenter}, node $\phi^{-1}(w) \in \mathrm{center}(G)$. It follows that the distance between $v$ and $\mathrm{center}(G)$ is at most $m'$, i.e., $m \leq m'$.
	
	Altogether, we have shown that $m = m'$, as desired.
\end{proof}


	Next, we prove a structural result about any non-trivial automorphism $\phi$ of a tree. At a high level, we prove that for any node $x$ not fixed by $\phi$, the path with endpoints $x$ and $\phi(x)$ exhibits a reflective symmetry about the center of the path.
\begin{lemma}\label{pathsymmetry}
	For any tree $G$ and any non-trivial automorphism $\phi \in \mathrm{Aut}(G)$, consider any node $x$ such that $x \neq \phi(x)$. Let $\ell \geq 1$ be the length of the path with endpoints $x$ and $\phi(x)$. Let $v_1 = x$ and let $v_{\ell+1} = \phi(x)$, and denote by   $(v_1,\ldots,v_{\ell+1})$ the sequence of nodes along the path. For each $i \in \{0,\ldots,\lfloor \ell/2 \rfloor\}$, we have that $v_{\ell+1-i} = \phi(v_{1+i})$.
\end{lemma}
\begin{proof}
We proceed by induction on $i$. The case where $i = 0$ is true since we defined $v_1 = x$ and $v_{\ell+1} = \phi(x)$. As induction hypothesis, assume that for some $i \in \{0,\ldots,\lfloor \ell/2 \rfloor-1\}$, we have that $v_{\ell+1-i} = \phi(v_{1+i})$. 

We proceed to show that $v_{\ell+1-(i+1)} = \phi(v_{1+(i+1)})$. By the induction hypothesis, we know that $v_{\ell+1-i} = \phi(v_{1+i})$, so, by Corollary \ref{preservedistcenter}, it follows that the distance between $v_{\ell+1-i}$ and $\mathrm{center}(G)$ is equal to the distance between $v_{1+i}$ and $\mathrm{center}(G)$. As $G$ is a tree, $\mathrm{center}(G)$ is either a single node or the two endpoints of a single edge. In either case, there is a unique path between $v_{1+i}$ and the node in $\mathrm{center}(G)$ closest to it. In particular, $v_{1+i}$ has a unique neighbour $w_1$ whose distance to $\mathrm{center}(G)$ is smaller than $v_{1+i}$'s distance to $\mathrm{center}(G)$. Similarly, there is a unique path between $v_{\ell+1-i}$ and the node in $\mathrm{center}(G)$ closest to it. In particular, $v_{\ell+1-i}$ has a unique neighbour $w_2$ whose distance to $\mathrm{center}(G)$ is smaller than $v_{\ell+1-i}$'s distance to $\mathrm{center}(G)$. Since the automorphism $\phi$ preserves adjacencies (Definition \ref{autodefs}) and preserves distance to $\mathrm{center}(G)$ (Corollary \ref{preservedistcenter}), we conclude that $\phi(w_1) = w_2$. 

Next, observe that $w_1$ is necessarily equal to $v_{1+(i+1)}$: indeed, if $v_{1+(i+1)} \neq w_1$, then, since $G$ is a tree, all nodes on the path $(v_{1+(i+1)},\ldots,v_{\ell+1-i},\ldots,v_{\ell+1})$ would have a distance to $\mathrm{center}(G)$ strictly greater than the distance between $v_{1+i}$ and $\mathrm{center}(G)$, which contradicts the fact that $v_{1+i}$ and $v_{\ell+1-i}$ have the same distance to $\mathrm{center}(G)$. Similarly, observe that $w_2$ is necessarily equal to $v_{\ell+1-(i+1)}$: indeed, if $v_{\ell+1-(i+1)} \neq w_2$, then, since $G$ is a tree, all nodes on the path $(v_{1},\ldots,v_{1+i},\ldots,v_{\ell+1-(i+1)})$ would have a distance to $\mathrm{center}(G)$ strictly greater than the distance between $v_{\ell+1-i}$ and $\mathrm{center}(G)$, which contradicts the fact that $v_{\ell+1-i}$ and $v_{1+i}$ have the same distance to $\mathrm{center}(G)$. As we have previously shown that $\phi(w_1) = w_2$, this concludes the proof that $v_{\ell+1-(i+1)} = \phi(v_{1+(i+1)})$.
\end{proof}

\begin{theorem}
	Consider any labeling scheme $\lambda$ and any deterministic distributed algorithm $\mathcal{A}$. If $\mathcal{A}$ solves the gossiping task when executed by the nodes of $\lambda(G)$ for some tree $G$, then the length of $\lambda$ is $\Omega(\log D(G))$.
\end{theorem}
\begin{proof}
	Consider any labeling scheme $\lambda$ and any deterministic distributed algorithm $\mathcal{A}$ that solves the gossiping task when executed by the nodes of $\lambda(G)$ for some tree $G$. We prove that $\lambda$ is a distinguishing labeling of $G$, which implies that the number of distinct labels used by $\lambda$ is at least $D(G)$, and this immediately implies the result. 
	
	To obtain a contradiction, assume that there exists a non-trivial $\phi \in \textrm{Aut}(G)$ such that $\phi$ preserves $\lambda$. By definition, this means that, for each node $x$, we have $\lambda(x) = \lambda(\phi(x))$. 
	
	Consider the execution of $\mathcal{A}$ on the labeled network $\lambda(G)$. First, we set out to prove that, for each node $x$ in the network, nodes $x$ and $\phi(x)$ behave exactly the same way in each round. To do so, we consider node histories during the execution of $\mathcal{A}$: for an arbitrary node $x$, define $h_x[0]$ to be the label assigned to $x$ by $\lambda$, and, for each $t \geq 1$, define $h_x[t]$ to be the message received by $x$ in round $t$ (or $\bot$ if $x$ receives no message). As $\mathcal{A}$ is a deterministic distributed algorithm, we know that $h_x[0\ldots (t-1)] = h_y[0\ldots (t-1)]$ implies that $x$ and $y$ perform the exact same action in round $t$, i.e., $x$ and $y$ both stay silent in round $t$, or, they both transmit the same message in round $t$.
	\begin{claim}
		For an arbitrary round $t \geq 1$ and an arbitrary node $x$, we have that $x$ transmits in round $t$ if and only if $\phi(x)$ transmits in round $t$. Further, if $x$ and $\phi(x)$ both transmit in round $t$, then they transmit the same message.
	\end{claim}
	To prove the claim, it suffices to prove that, for each $t \geq 1$ and each node $x$, nodes $x$ and $\phi(x)$ have the same history up to round $t-1$, i.e., $h_x[0\ldots (t-1)] = h_{\phi(x)}[0 \ldots (t-1)]$. We proceed by induction on the round number $t$. For the base case, consider $t=1$. By assumption, we have that $\lambda(x) = \lambda(\phi(x))$ for each node $x$, which, by definition, gives $h_x[0] = h_{\phi(x)}[0]$ for each node $x$, as required. As induction hypothesis, assume that, for some $t \geq 1$, that $h_x[0\ldots (t-1)] = h_{\phi(x)}[0 \ldots (t-1)]$ for each node $x$. For the inductive step, consider an arbitrary node $x$, and consider the possible cases for the value of $h_x[t]$:
	\begin{itemize}
		\item {\bf Suppose that $h_x[t] = \bot$.} There are several sub-cases to consider:
		\begin{itemize}
			\item {\bf $x$ transmits in round $t$.} Then, by the induction hypothesis, we know that $h_x[0\ldots (t-1)] = h_{\phi(x)}[0 \ldots (t-1)]$, which implies that $\phi(x)$ transmits in round $t$ as well. This means that $\phi(x)$ does not receive a transmission in round $t$, so $h_{\phi(x)}[t] = \bot$, as required.
			\item {\bf $x$ does not transmit in round $t$, and no neighbours of $x$ transmit in round $t$.} Then, by the definition of automorphism, a node $y$ is a neighbour of $x$ if and only if node $\phi(y)$ is a neighbour of $\phi(x)$. In particular, each neighbour of $\phi(x)$ is some $\phi(y)$ where $y$ is a neighbour of $x$. By the induction hypothesis, we know that $h_y[0\ldots (t-1)] = h_{\phi(y)}[0 \ldots (t-1)]$ for each neighbour $y$ of $x$, so $\phi(y)$ and $y$ perform the same action in round $t$. By assumption, no neighbour $y$ of $x$ transmits in round $t$, so we get that no neighbour $\phi(y)$ of $\phi(x)$ transmits in round $t$, which implies that $h_{\phi(x)}[t] = \bot$, as required.
			\item {\bf $x$ does not transmit in round $t$, and two or more neighbours of $x$ transmit in round $t$.} Let $y_1, y_2$ be two distinct neighbours of $x$ that transmit in round $t$. By the definition of automorphism, $\phi(y_1)$ and $\phi(y_2)$ are neighbours of $\phi(x)$. By the induction hypothesis, we know that $h_{y_1}[0\ldots (t-1)] = h_{\phi(y_1)}[0 \ldots (t-1)]$ and $h_{y_2}[0\ldots (t-1)] = h_{\phi(y_2)}[0 \ldots (t-1)]$. In particular, this implies that $\phi(y_1)$ and $\phi(y_2)$ both transmit in round $t$, which causes a collision at $\phi(x)$ in round $t$. Thus, $\phi(x)$ does not receive a message in round $t$, i.e., $h_{\phi(x)}[t] = \bot$, as required.
		\end{itemize}
		\item {\bf Suppose that $h_x[t] = m$ for some binary string $m$.} In particular, this means that exactly one neighbour $y'$ of $x$ transmits in round $t$, and $y'$ sends binary string $m$ in its transmission. By the definition of automorphism, a node $y$ is a neighbour of $x$ if and only if node $\phi(y)$ is a neighbour of $\phi(x)$. In particular, each neighbour of $\phi(x)$ is some $\phi(y)$ where $y$ is a neighbour of $x$. By the induction hypothesis, we know that $h_y[0\ldots (t-1)] = h_{\phi(y)}[0 \ldots (t-1)]$ for each neighbour $y$ of $x$, so $\phi(y)$ and $y$ perform the same action in round $t$. By assumption, $y'$ transmits a message $m$ during round $t$, so $\phi(y')$ transmits message $m$ during round $t$. Each neighbour $y \neq y'$ of $x$ is silent in round $t$, so each neighbour $\phi(y) \neq \phi(y')$ of $\phi(x)$ is silent in round $t$. Thus, $\phi(x)$ receives message $m$ in round $t$, i.e., $h_{\phi(x)}[t] = m$, as required.
	\end{itemize}
	In all cases, we showed that $h_x[t] = h_{\phi(x)}[t]$ for an arbitrary node $x$, which, together with the induction hypothesis, proves that $h_x[0\ldots t] = h_{\phi(x)}[0\ldots t]$ for all nodes $x$. This completes the proof of the claim.

	Next, we use the above claim to reach a contradiction: that $\mathcal{A}$ does not solve the gossiping task. We consider a node $x$ such that $x \neq \phi(x)$ (which must exist since $\phi$ is a non-trivial automorphism). We will prove that $x$'s source message does not reach $\phi(x)$, which is sufficient to prove that the gossiping task is not correctly solved. (Note: this is the only place in the proof that needs to be modified for the result to apply more generally to $k$-broadcast with $k \in \{2,\ldots,n\}$ when the sources are not initially known: after the labeling scheme is applied, choose the $x$ and $\phi(x)$ specified above as two of the $k$ sources.)
	
	Consider the path $P$ with endpoints $x$ and $\phi(x)$, and denote the length of this path by $\ell$. Let $v_1 = x$ and let $v_{\ell+1} = \phi(x)$, and denote by $(v_1,\ldots,v_{\ell+1})$ the sequence of nodes along the path. We consider two cases based on whether or not $\ell$ is even or odd.
	\begin{itemize}
		\item \underline{\bf Case 1: $\ell$ is even.} In this case, the path $P$ has a central node $v_{1+(\ell/2)}$ (i.e., a node with the same distance to both endpoints of $P$). Consider the two neighbours of $v_{1+(\ell/2)}$ on $P$, i.e., nodes $v_{\ell/2}$ and $v_{2+(\ell/2)}$. By Lemma \ref{pathsymmetry} with $i = (\ell/2)-1$, it follows that $v_{2+(\ell/2)} = \phi(v_{\ell/2})$. By the above claim, for every round in the execution of $\mathcal{A}$, both $v_{\ell/2}$ and $v_{2+(\ell/2)}$ perform the same action, i.e., either both transmit or both listen. It follows that, in every round, $v_{1+(\ell/2)}$ does not receive a message from either $v_{\ell/2}$ or $v_{2+(\ell/2)}$, either due to transmission collision or due to both remaining silent.
		\item \underline{\bf Case 2: the length of $P$ is odd.} In this case, the path $P$ has a central edge with endpoints $v_{1+\lfloor \ell/2 \rfloor}$ and $v_{\ell+1-\lfloor \ell/2 \rfloor}$ (i.e., an edge whose two endpoints have same distance to their closest endpoint of $P$). By Lemma \ref{pathsymmetry} with $i = \lfloor \ell/2 \rfloor$, it follows that $v_{\ell+1-\lfloor \ell/2 \rfloor} = \phi{v_{1+\lfloor \ell/2 \rfloor}}$. By the above claim, for every round in the execution of $\mathcal{A}$, both $v_{1+\lfloor \ell/2 \rfloor}$ and $v_{\ell+1-\lfloor \ell/2 \rfloor}$ perform the same action, i.e., either both transmit or both listen. It follows that these two nodes never receive a message from each other.
	\end{itemize}
As $G$ is a tree, the only way that $x$'s source message can reach $\phi(x)$ is if, for each $j \in \{1,\ldots,\ell\}$, there exists a round in which node $v_{j+1}$ receives $x$'s source message from $v_j$. In both cases above, we demonstrated a value of $j$ for which this never occurs: $j=\ell/2$ in Case 1, and $j=\ell-\lfloor \ell/2 \rfloor$ in Case 2. Since $\phi(x)$ never receives $x$'s source message, we have shown that $\mathcal{A}$ does not solve the gossiping task, which proves the desired contradiction.
\end{proof}

\subsection{Upper Bound}
We present a $O(\log D(G))$-bit labeling scheme and a deterministic distributed gossiping algorithm that runs on the labeled network. The overall idea is similar to our labeling scheme and algorithm for $k$-broadcast in Section \ref{arbkbroad}. With regards to the labeling scheme: the nodes are labeled according to $\ackscheme$, the labeling scheme arbitrarily chooses a coordinator node $r$ and uniquely labels it as such, and then sets additional $sched$ bits according to a particular graph colouring (however, in the present case, a distinguishing labeling (see Definition \ref{distinguishdefs}) is used rather than a distance-two colouring). With regards to the algorithm, the execution consists of three subroutines performed consecutively: {\bf Initialize}, {\bf Aggregate}, and {\bf Inform}. In the first stage, using bounded acknowledged broadcast, each node learns its distance from the coordinator as well as the eccentricity of the coordinator. In the second stage, all of the source messages are collected at the coordinator node, although the procedure to do so is significantly more complex (as avoiding collisions is not as easy now that we are using a distinguishing labeling rather than a distance-two colouring). At a high level, each node will use its colour, its distance from the coordinator, and messages from its children in the tree to create an `encoding' of its rooted subtree, and we can prove that this is necessarily different than the encoding calculated by its siblings in the tree (due to the distinguishing labeling). Using these distinct encodings, the nodes choose distinct delay values that are mutually co-prime, and repeated transmissions separated by these delay values will guarantee that their messages will eventually be received by their parent in the tree. Finally, in the third stage, the coordinator shares the set of source messages with the entire network using a final broadcast.

\subsubsection{Labeling Scheme $\gossipscheme$}
As input, we are provided with a graph $G$. We assign a label to each node $v$ in $G$, and the label at each node $v$ consists of 5 components: a $join$ bit, a $stay$ bit, an $ack$ bit, a $term$ bit, and a binary string $sched$. The labeling scheme assigns values to the components as follows:
\begin{enumerate}
	\item Choose an arbitrary node $r \in G$. This node will act as the \emph{coordinator}.
	\item Apply the labeling scheme $\ackscheme$ (see Section \ref{EGMPAckBroadcast}) to $G$ with designated start node $s_G=r$. This will set the $join$, $stay$, and $ack$ components (each consisting of one bit) at each node $v$. For the coordinator node $r$, set $join=stay=ack=1$.
	\item For each node $v$, set the $sched$ bits to be the binary representation of $\psi(v)$, where $\psi$ is any $D(G)$-distinguishing labeling of the nodes of $G$.
	\item The $term$ bit is set to 1 at exactly one node: the node whose source message is last to arrive at the coordinator $r$ during the {\bf Aggregate} subroutine (as described at the end of Section \ref{defineGossipAlg}). The $term$ bit at all other nodes is set to 0.
\end{enumerate}

\subsubsection{Gossiping Algorithm}\label{defineGossipAlg}
We now describe our deterministic distributed gossiping algorithm $\mathcal{GOSSIP}$ that is executed after the tree's nodes have been labeled using labeling scheme $\gossipscheme$. The algorithm's execution consists of three subroutines performed consecutively: {\bf Initialize}, {\bf Aggregate}, and {\bf Inform}.

The first stage of the algorithm, consisting of the {\bf Initialize} subroutine, consists of executing $\mathcal{B}_{bounded}$. The start node is the coordinator $r$ (the unique node with $join$, $stay$, and $ack$ bits all set to 1), and the broadcast message is ``init" along with a \texttt{counter} value that is 0 in the coordinator's initial message. During the execution of the initial broadcast, each node increments this counter value before re-transmitting the broadcast message. As $G$ is a tree, no collisions occur during this execution, so the first received \texttt{counter} value by a node $v$ is actually $v$'s distance from $r$. Further, each node can set its local clock to this received \texttt{counter} value in order to establish a global clock. At the conclusion of the execution of $\mathcal{B}_{bounded}$, all nodes know an upper bound $m$ on the number of rounds that elapsed during the broadcast of the ``init" message, and, they all received this value before round $t_{done} = 3m$. Thus, in round $3m$, all nodes terminate the subroutine. Once again, as $G$ is a tree and no collisions occur during the execution of $\mathcal{B}_{bounded}$, the value of $m$ is equal to the exact height of the tree rooted at $r$, since this is exactly how many rounds it took for the initial broadcast to complete.

The third stage of the algorithm, consisting of the {\bf Inform} subroutine, consists of executing $\mathcal{B}_{ack}$. The start node is the coordinator $r$, and the broadcast message is equal to the set of source messages that $r$ knows. All nodes terminate this subroutine at the same time, and they all know that gossiping has been completed.

The second stage of the algorithm, mainly consisting of the {\bf Aggregate} subroutine, is the most interesting. The remainder of this section is dedicated to its description and analysis.

Consider the tree $G$ rooted at the coordinator node $r$. At the start of this stage, each node $v$ knows its distance from $r$, which we denote by $d(v)$. Also, each node knows the height $m$ of the tree. The main idea is for each node to compute an encoding of the subtree rooted at itself, and then use this encoding to determine a transmission delay value that it will use in order to avoid transmission collisions when sending information to its parent.

Each node $v$ maintains an encoding enc($v$) in every round $t$, and, in rounds in which $v$ decides to transmit, it always includes this enc($v$) value in the transmitted message. If $v$ has not yet received any messages from any children (which is necessarily the case in the first round, and, in all rounds when $v$ is a leaf), then node $v$ computes its encoding using the distance $d(v)$ from the coordinator and its colour $\psi(v)$ (which is stored in the $sched$ bits of its label). More specifically, the encoding is an integer computed as $\textrm{enc}(v) = 2^{\psi(v)}\cdot 3^{d(v)}$. Next, suppose that a node $w$ has received messages from some subset $\{v_1,\ldots,v_\ell\}$ of its children. Then $w$ computes its encoding using $d(w)$, its colour $\psi(w)$, and the encodings of its children. In particular, it computes $enc(w) = 2^{\psi(w)}\cdot 3^{d(w)} \cdot \prod_{j=1}^{\ell} p_{j+2}^{\textrm{enc}(v_j)}$, where $p_{j+2}$ denotes the $(j+2)^{\textrm{th}}$ smallest prime number. For any two nodes $a,b$ with a common ancestor $c$, we will be able to show that this method of encoding ensures that $\textrm{enc}(a), \textrm{enc}(b), \textrm{enc}(c)$ are all different: either because $d(c)$ is different from $d(a),d(b)$, or, because $\psi(a) \neq \psi(b)$, or, because the subtrees rooted at $a$ and $b$ are non-isomorphic (which must be the case if $\psi(a) = \psi(b)$, by the definition of distinguishing colouring). 

Using its computed $\textrm{enc}(v)$ value, each node $v$ computes a transmission delay value $\tau(v)$. In particular, it sets $\tau(v) = p_{\textrm{enc}(v)}$ (once again, using $p_i$ to denote the $i^{\textrm{th}}$ prime number). Then, $v$ transmits every $\tau(v)$ rounds, and in its transmitted message, it includes: the value of $\textrm{enc}(v)$, the set of source messages it knows, and its distance from $r$ (so that $v$'s parent can recognize that the message is coming from one of its children). Using the fact that siblings $a,b$ and any common ancestor $c$ will have different $\textrm{enc}(a), \textrm{enc}(b),\textrm{enc}(c)$, we will be able to show that all siblings and their common ancestors will have different transmission delay values that are mutually co-prime. This implies that each sibling will eventually successfully transmit its knowledge to its parent in $G$.

One challenge that was not addressed in the above description is that a node $v$ will receive messages from different children at different times, i.e., the values $\textrm{enc}(v)$ and $\tau(v)$ are computed in each round using only the information that $v$ has received from its children before the current round. There are three potential issues: (1) $v$ may have children that it has not heard from yet (and it has no way of detecting this); (2) $v$ may have one or more children $x$ from which it has received incomplete information (e.g., $x$ may have children that it has not heard from yet); (3) $v$ may hear from a child $x$ multiple times with different encoding values (e.g., $x$ updates its own encoding between transmissions). To deal with issue (3), when we receive an encoding enc($x$) from a child $x$: we replace a previously saved encoding $e$ if enc($x$) is a multiple of $e$, and otherwise we append enc($x$) to our list of saved encodings. Issues (1) and (2) demonstrate that the guarantees claimed previously about distinct encoding values (and distinct transmission delay values) do not necessarily hold at all times. We are able to carefully prove that each of these issues is correctly handled, i.e., that node $v$ \emph{eventually} receives a message from each of its children and has complete information about its entire subtree, and from that point on, the transmission delays will diverge and allow siblings to successfully transmit up to their parent.

Algorithm \ref{aggpseudo} provides a pseudocode description of the {\bf Aggregate} subroutine. 

\begin{algorithm}[H]
	\small
	\caption{The {\bf Aggregate} subroutine executed at each node $v$}
	\label{aggpseudo}
	\begin{algorithmic}[1]
		\Statex {\color{gray} \% Each node has a source message $\mu_v$. Each node knows its distance $d(v)$ from the coordinator $r$. In its label, each node has $\mathit{sched}$ bits: the binary representation of its colour $\psi(v)$}
		\State $\texttt{knownMsgs} \leftarrow \{\mu_v\}$
		\State $\texttt{colour} \leftarrow \textrm{integer value of $\mathit{sched}$ bits}$
		\State $\texttt{dist} \leftarrow d(v)$ \Comment{{\color{gray}my distance to coordinator}}
		\State $\texttt{enc} \leftarrow 2^{\tt colour} * 3^{\tt dist}$ \Comment{{\color{gray}encoding of my subtree}}
		\State $\texttt{childrenEncs} \leftarrow$ empty list \Comment{{\color{gray}will store encodings from children}}
		\State $\texttt{delay} \leftarrow \mathit{prime}(\texttt{enc})$ \Comment{{\color{gray}$\mathit{prime}(n)$ returns the $n^{\textrm{th}}$ prime number}}
		\For{each round $t \geq 1$}
		\If{($\texttt{dist} > 0$) {\bf and} ($t\ \textbf{mod}\ \texttt{delay} = 0$)}
		\Statex {\color{gray} \algindent\algindent\% non-coordinator nodes transmit every \texttt{delay} rounds}
		\State transmit $\langle \texttt{knownMsgs}, \texttt{enc}, \texttt{dist} \rangle$
		\Else 
		\State listen for a message
		\If{received $\langle \texttt{recvMessages}, \texttt{recvEnc}, \texttt{recvDist} \rangle$}
		\If{$\texttt{dist} = \texttt{recvDist}-1$}
		\Statex \algindent\algindent\algindent\algindent {\color{gray}\% I am the parent of this message's transmitter
			\Statex \algindent\algindent\algindent\algindent\% First, save the received source messages}
		\State $\texttt{knownMsgs} \leftarrow \texttt{knownMsgs} \cup \texttt{recvMessages}$
		\Statex \algindent\algindent\algindent\algindent {\color{gray}\% replace previously saved encoding or append}
		\If{$\exists i, \texttt{recvEnc} \textbf{ mod } \texttt{childrenEncs}[i] = 0$}
		\State $\ell \leftarrow \min\{i\ |\ \texttt{recvEnc} \textbf{ mod } \texttt{childrenEncs}[i] = 0\}$
		\State $\texttt{childrenEncs}[\ell] \leftarrow \texttt{recvEnc}$
		\Else
		\State append \texttt{recvEnc} to \texttt{childrenEncs}
		\EndIf
		\Statex \algindent\algindent\algindent\algindent {\color{gray}\% update my encoding and delay values}
		\State $\texttt{nc} \leftarrow |\texttt{childrenEncs}|$ \Comment{{\color{gray}\# of encodings from children}}
		\State $\texttt{enc} \leftarrow 2^{\tt colour} * 3^{\tt dist} * \displaystyle\prod_{j=1}^{\texttt{nc}} [\mathit{prime}(j+2)] ^ {\texttt{childrenEncs}[j]}$ 
		\State $\texttt{delay} \leftarrow \mathit{prime}(\texttt{enc})$
		
		\EndIf 
		\EndIf
		
		\EndIf
		\EndFor
	\end{algorithmic}
\end{algorithm}

To prove the correctness of the {\bf Aggregate} subroutine, we set out to prove that, for each node $w$ in the tree $G$ rooted at $r$, there exists a round $t_w$ such that $w$ has received all of the source messages originating at nodes in the subtree rooted at $w$. Taking $v=r$, this proves that the coordinator node eventually possesses all of the source messages. 

In what follows, for any node $w$, recall that the \emph{height} of $w$ is the length of the longest root-to-leaf path in the subtree rooted at $w$. For any variable $V$ used in the pseudocode of Algorithm \ref{aggpseudo}, we write $V_t(w)$ to represent ``the value of $w$'s variable $V$ at the start of round $t$".

First, we observe that nodes at different distances from the coordinator will have different \texttt{enc} values at all times, as these distances are used as an exponent in the prime decomposition of the encoding value. Based on how the \texttt{delay} values are calculated, we can also conclude that nodes at different distances from the coordinator will have different \texttt{delay} values at all times. This fact will be used later to help ensure that there is an opportunity for each node to successfully transmit to its parent by avoiding interfering transmissions from its parent or grandparent in the tree.
\begin{proposition}\label{ancestordifferent}
	Consider any distinct nodes $v,w$ and any round $t$. If $w$ is an ancestor of $v$, then $\texttt{delay}_t(v) \neq \texttt{delay}_t(w)$.
\end{proposition}
\begin{proof}
	Suppose that $w$ is an ancestor of $v$. This implies that $d(w) < d(v)$, i.e., the distance from $w$ to $r$ is strictly smaller than the distance from $v$ to $r$. However, this means that in the calculations of $\texttt{enc}_t(v)$ and $\texttt{enc}_t(w)$, the exponent of 3 is different, which implies that $\texttt{enc}_t(v) \neq \texttt{enc}_t(w)$, and thus $\texttt{delay}_t(v) \neq \texttt{delay}_t(w)$.
\end{proof}

Next, we consider the issue that nodes may perform transmissions before they have complete information about all nodes in their subtree, so their encoding of their subtree might change over time. When a node transmits an updated value of its subtree encoding, we need a way for the parent to distinguish whether the information it receives is from a child it has not heard from before, or if the information should overwrite a value that was sent previously by one of its children. The following result will give us a tool that will be useful in solving this issue.
\begin{lemma}\label{divides}
	For any node $v$ in $G$, we have that $\texttt{enc}_t(v)$ divides $\texttt{enc}_{t'}(v)$ for all $t,t'$ such that $t' \geq t$.
\end{lemma}
\begin{proof}
	For any two consecutive rounds for which $\texttt{enc}_t(v) = \texttt{enc}_{t+1}(v)$, it follows immediately that $\texttt{enc}_t(v)$ divides $\texttt{enc}_{t+1}(v)$. So the remainder of the proof considers any round $t$ such that the value of $\texttt{enc}$ changes in round $t$, i.e., $\texttt{enc}_t(v) \neq \texttt{enc}_{t+1}(v)$. From the definition of $\texttt{enc}$, we can write $\texttt{enc}_t(v)$ as a product of prime powers $p_1^{\alpha_1}p_2^{\alpha_2}\cdots p_\ell^{\alpha_\ell}$ for some $\ell \geq 2$. From the definition of {\bf Aggregate}, we see that there are two ways that the value of \texttt{enc} can change in a round: (1) a value is appended to the \texttt{childrenEncs} list; or, (2) a value in the \texttt{childrenEncs} list is overwritten. In case (1), suppose that a new element $e$ is appended to \texttt{childrenEncs}. Then, the value $\texttt{enc}_{t+1}(v)$ is calculated as $p_1^{\alpha_1}p_2^{\alpha_2}\cdots p_\ell^{\alpha_\ell}\cdot p_{\ell+1}^{e} = \texttt{enc}_t(v) \cdot p_{\ell+1}^{e}$, which implies that $\texttt{enc}_t(v)$ divides $\texttt{enc}_{t+1}(v)$. In case (2), suppose that a value in the \texttt{childrenEncs} list is overwritten, say, at position $i$ in the list. This can only occur when the if condition on line 15 evaluates to true, i.e., the new value written at position $i$ is a multiple of the old value, say, $c \cdot \alpha_i$. Then, the value $\texttt{enc}_{t+1}(v)$ is calculated as $p_1^{\alpha_1}p_2^{\alpha_2}\cdots p_i^{c \cdot \alpha_i} \cdots p_\ell^{\alpha_\ell} = p_i^{(c-1)\cdot \alpha_i}\cdot [p_1^{\alpha_1}p_2^{\alpha_2}\cdots p_i^{\alpha_i} \cdots p_\ell^{\alpha_\ell}] = p_i^{(c-1)\cdot \alpha_i}\cdot \texttt{enc}_t(v)$, which implies that $\texttt{enc}_t(v)$ divides $\texttt{enc}_{t+1}(v)$.
\end{proof}
The previous result implies that every time a fixed node transmits, the $\texttt{enc}$ value in the most recent transmission is always a multiple of previous $\texttt{enc}$ values it has transmitted. This motivates the strategy used at lines 15-20 in {\bf Aggregate}: if an \texttt{enc} value is received that is a multiple of an older value, the older value gets overwritten, and otherwise the new value is appended to the list. This gives us the next result, which tells us that the size of the $\texttt{childrenEncs}$ list at any node $v$ will never be larger than the number of children that $v$ has. We will use this fact later when proving that, eventually, \texttt{childrenEncs} consists exactly of the \texttt{enc} values of all of $v$'s children.
\begin{corollary}\label{childrenlimit}
	Consider any node $v$ in $G$. If $v$ has $c$ children, then $|\texttt{childrenEncs}_t(v)| \leq c$ in every round $t$. 
\end{corollary}
\begin{proof}
	To obtain a contradiction, assume there is a time $t$ such that $|\texttt{childrenEncs}_t(v)| > c$. As elements are only added to\\ \texttt{childrenEncs} at line 19, which only occurs after receiving a message from a child in $G$, it follows that the only elements stored in \texttt{childrenEncs} are \texttt{enc} values received from children. By the Pigeonhole Principle, it follows that there are two elements, say $\alpha,\beta$, in \texttt{childrenEncs} at time $t$ that were received in messages from the same child. Without loss of generality, assume that $\alpha$ was received before $\beta$. However, by Lemma \ref{divides}, we know that $\alpha$ divides $\beta$, i.e., $\beta$ is a multiple of $\alpha$. It follows that, when $\beta$ is received, the if condition at line 15 will evaluate to true and $\beta$ will replace $\alpha$ in \texttt{childrenEncs}, contradicting the fact that $\alpha$ and $\beta$ are both in \texttt{childrenEncs} at time $t$.
\end{proof}

As mentioned earlier, there might be some initial period in the execution of {\bf Aggregate} during which nodes compute and transmit incomplete information about their subtree, so our desired guarantees about distinct encodings and distinct transmission delays do not necessarily hold at all times. The following definition will help us talk about the eventual behaviour of the subroutine.
\begin{definition}
	We say that a node $w$ \emph{is settled in round $t$} if, at the start of round $t$, node $w$ possesses the source messages of all nodes in the subtree rooted at $w$, and, $\texttt{delay}_t(w) = \texttt{delay}_{t'}(w)$ for all $t' > t$. Note that, by the way $\texttt{delay}$ is calculated, this last condition is equivalent to saying $\texttt{enc}_t(w) = \texttt{enc}_{t'}(w)$ for all $t' > t$.
\end{definition}
Earlier, we proved that nodes at different distances from the coordinator will have different transmission delay values. We want to prove something similar about sibling nodes, i.e., any two nodes with the same parent, so that we can guarantee that each sibling has an opportunity to successfully transmit to the parent without interference. The desired result is a corollary of the following lemma.
\begin{lemma}\label{iso}
	Consider any distinct nodes $v,w$ in $G$ and a round $t$ such that: (i) all nodes in the subtrees rooted at $v$ and $w$ have settled at time $t$, and, (ii) $\texttt{enc}_t(v) = \texttt{enc}_t(w)$. Then the subtree rooted at $v$ can be mapped to the subtree rooted at $w$ via a colour-preserving isomorphism.
\end{lemma}
\begin{proof}
	We proceed by induction on the height of the subtree rooted at $v$. For the base case, suppose that the height of the subtree rooted at $v$ is 0. Then $v$ is a leaf node, and the value of $\texttt{enc}_t(v)$ in all rounds is equal to $2^{colour(v)} \cdot 3^{d(v)}$. Since $\texttt{enc}_t(v) = \texttt{enc}_t(w)$, we conclude that the value of $\texttt{enc}_t(w)$ in all rounds after $t$ is also $2^{colour(v)} \cdot 3^{d(v)}$, which implies that $w$ has no children (otherwise there would be larger primes in the product), and $v$ and $w$ were assigned the same colour by $\psi$. It follows that the subtrees rooted at $v$ and $w$ are both single nodes, and they were assigned the same colour by $\psi$, so the function that maps $v$ and $w$ to each other is a colour-preserving isomorphism. As induction hypothesis, assume that the statement of the lemma holds for nodes with any height at most $h-1$ for some $h \geq 1$. For the inductive step, we suppose that the subtree rooted at some node $v$ has height exactly $h$, we consider a node $w$ and a time $t$ such that all nodes in the subtrees rooted at $v$ and $w$ have settled at time $t$, and, we suppose that $\texttt{enc}_t(v) = \texttt{enc}_t(w)$. Denote by $V = \{v_1,\ldots,v_{|V|}\}$ the children of $v$, and denote by $W = \{w_1,\ldots,w_{|W|}\}$ the children of $w$, where we have assigned the subscripts of the children in increasing order by the child's \texttt{enc} value at round $t$. By the calculation of \texttt{enc}, it follows that\\
	$2^{colour(v)}\cdot 3^{d(v)} \cdot \displaystyle\prod_{i=1}^{|V|} p_{i+2}^{\texttt{enc}_t(v_i)} = 2^{colour(w)}\cdot 3^{d(w)} \cdot \displaystyle\prod_{j=1}^{|W|} p_{j+2}^{\texttt{enc}_t(w_j)}$. It follows that $colour(v) = colour(w)$, $d(v)=d(w)$, $|V| = |W|$, and, for each $i \in \{1,\ldots,|V|\}$, we have $\texttt{enc}_t(v_i) = \texttt{enc}_t(w_i)$. Applying the induction hypothesis to the subtrees rooted at each pair $v_i,w_i$ (which both have height at most $h-1$) we conclude that there is a colour-preserving isomorphism from the subtree rooted at $v_i$ to the subtree rooted at $w_i$. Combining these isomorphisms and mapping $v$ to $w$, we get a colour-preserving isomorphism from the subtree rooted at $v$ to the subtree rooted at $w$.
\end{proof}
\begin{corollary}\label{siblingdifferent}
	Consider any two nodes $a,b$ with the same parent in $G$. If nodes $a$ and $b$ are settled in some round $t$, then $\texttt{delay}_{t'}(a) \neq \texttt{delay}_{t'}(b)$ (equivalently, $\texttt{enc}_{t'}(a) \neq \texttt{enc}_{t'}(b)$) for all rounds $t' \geq t$.
\end{corollary}
\begin{proof}
	To obtain a contradiction, assume that $\texttt{enc}_{t'}(a) = \texttt{enc}_{t'}(b)$ for all rounds $t' \geq t$. Then, by Lemma \ref{iso}, the subtree rooted at $a$ can be mapped to the subtree rooted at $b$ using a colour-preserving isomorphism. For all other nodes $v$ in $G$, map $v$ to itself. Altogether, this gives a colour-preserving automorphism of $G$, which contradicts the fact that the colours were assigned to the nodes of $G$ using a distinguishing colouring.
\end{proof}
Finally, we are ready to prove the correctness of {\bf Aggregate}. The essence of the proof is to inductively show, from the leaves upward towards the coordinator, that all nodes eventually settle, i.e., they eventually calculate a fixed transmission delay value, and eventually receive all of the source messages contained in their subtree. The important fact that enables this to happen is that the transmission delay values of all nodes with a common parent, along with the transmission delay values of the parent and the grandparent (if one exists), are all different. Further, the fact that these values are all prime numbers makes it easy to see that, by repeatedly transmitting with the same delay value, each of the sibling nodes will eventually have a chance to transmit successfully to the common parent.
\begin{theorem}\label{wsettles}
	For each node $w$ in $G$, there exists a round $t_w$ such that node $w$ is settled in round $t_w$.
\end{theorem}
\begin{proof}
	The proof proceeds by induction on the height of $w$. For the base case, consider any node $w$ with height 0. Then $w$ is a leaf and never receives any messages from any children. It follows that $w$ never executes the lines 14-22 in {\bf Aggregate}, so its \texttt{delay} value is not modified in any round to a value other than its initial value. Further, as $w$ starts with its own source message $\mu_w$, it follows that $w$ possesses the source messages of all nodes in the subtree rooted at $w$. This concludes the proof that $w$ is settled in round 1 in the case where $w$'s height is 0. As induction hypothesis, for some fixed $h \geq 1$, assume that for each node $v$ with height strictly less than $h$, there exists a round $t_v$ such that $v$ is settled in round $t_v$. For the inductive step, we now consider any node $w$ with height exactly $h$. 
	
	First, we consider the case where $w$ has exactly one child $v$. As the height of $v$ is $h-1$, the induction hypothesis implies that there is some round $t_v$ such that $v$ is settled in round $t_v$. After round $t_v$, node $v$ transmits in each round number divisible by $\texttt{delay}_{t_v}(v)$. By Proposition \ref{ancestordifferent}, node $w$ and $w$'s parent (if $w$ has one) have a \texttt{delay} value that is different from (and co-prime with) $\texttt{delay}_{t_v}(v)$ in all rounds after round $t_v$. It follows that there is a round $t'$ after $t_v$ in which $v$ transmits and both $w$ and $w$'s parent (if $w$ has one) listen, so $w$ will receive $v$'s transmission in round $t'$. But, since $v$ is settled in round $t_v$, it follows that $v$'s transmission in round $t'$ contains all of the source messages in the subtree rooted at $v$, so, after round $t'$, node $w$ possesses all of the source messages of nodes in the subtree rooted at $w$. Further, since $v$ is settled in round $t_v$, it follows that $v$'s transmission in every round after $t_v$ (including round $t'$) contains $\texttt{enc}_{t_v}(v)$, so every calculation of $\texttt{delay}$ at node $w$ after round $t'$ will be the same, which concludes the proof that node $w$ settles in round $t'$.
	
	Next, we consider the case where $w$ has at least two children, and we consider an arbitrary pair $a,b$ of $w$'s children. As the heights of $a$ and $b$ are both at most $h-1$, the induction hypothesis implies that there exists some round $t_a$ such that $a$ is settled in round $t_a$, and there exists some round $t_b$ such that $b$ is settled in round $t_b$. Let $t = \max\{t_a,t_b\}$, and note, by definition, that both $a$ and $b$ are settled in round $t$. By Lemma \ref{siblingdifferent}, it follows that $\texttt{delay}_{t'}(a) \neq \texttt{delay}_{t'}(b)$ (and $\texttt{enc}_{t'}(a) \neq \texttt{enc}_{t'}(b)$) for all rounds $t' \geq t$. Moreover, by Proposition \ref{ancestordifferent}, node $w$ and $w$'s parent (if $w$ has one) have a \texttt{delay} value that is different from $\texttt{delay}_{t'}(a)$ and $\texttt{delay}_{t'}(b)$ in all rounds $t' \geq t$. As all of these \texttt{delay} values are distinct prime numbers, and each node transmits when the current round number is divisible by \texttt{delay}, it follows that there is a round after $t$ in which $a$ transmits and all of $b$, $w$, and $w$'s parent (if it has one) listen, and a round after $t$ in which $b$ transmits and all of $a$, $w$, and $w$'s parent (if it has one) listen. It follows that there is some round after $t$ such that $w$ possesses all of the source messages that originated in the subtree rooted at $a$ and in the subtree rooted at $b$. It remains to show that there is some round $t''$ such that the \texttt{delay} value at $w$ remains the same in all rounds after $t''$.
	
	Recall that node $a$ is settled in round $t_a$, which implies that node $a$'s \texttt{enc} variable has some fixed value in all rounds after $t_a$. We denote this value by $\alpha$. Similarly, we denote by $\beta$ the fixed value of node $b$'s \texttt{enc} variable in all rounds after $t_b$. As shown above, we know that $\alpha \neq \beta$. We will prove below that, eventually, both $\alpha$ and $\beta$ are stored permanently as separate entries in node $w$'s \texttt{childrenEncs} list. As the argument below works for arbitrary children $a,b$ of $w$, it follows from Corollary \ref{childrenlimit} that there exists a round $t''$ after which node $w$'s \texttt{childrenEncs} list consists exactly of the settled \texttt{enc} values of all of $w$'s children. Thus, after round $t''$, the \texttt{delay} value at $w$ will remain fixed, which concludes the proof that $w$ is settled in round $t''$.
	
	The remainder of the proof is dedicated to proving that $\alpha$ and $\beta$ are eventually stored permanently as separate entries in node $w$'s \texttt{childrenEncs} list. There are three possible cases: (1) $\alpha$ divides $\beta$; (2) $\beta$ divides $\alpha$; (3) neither $\alpha$ or $\beta$ divides the other. In case (3), we see by lines 15-20 of {\bf Aggregate} that $\alpha$ and $\beta$ will be appended to \texttt{childrenEncs} and not replace one another. Cases (1) and (2) are symmetric, so, without loss of generality, we assume that $\alpha$ divides $\beta$ and that $\beta$ does not divide $\alpha$. First, we consider the case where $\beta$ is appended to \texttt{childrenEncs} before $\alpha$. In all future rounds, $\beta$ will not be overwritten by $\alpha$ since $\alpha$ is not a multiple of $\beta$ (i.e., the if condition on line 15 evaluates to false). Further, $\alpha$ will not be overwritten by $\beta$ since $\alpha$ is not located at the minimum index that contains an element that divides $\beta$ (in particular, $\beta$ itself is located at a smaller index than $\alpha$). So, in the case where $\beta$ is appended before $\alpha$, both elements will be appended to \texttt{childrenEncs} at $w$ and will not overwrite one another. Finally, consider the case where $\alpha$ is appended to \texttt{childrenEncs} before $\beta$. When $\beta$ is received by $w$, the if condition at line 15 will evaluate to true, so it is possible that $\alpha$ is overwritten by $\beta$ in the list. However, there will be a future round in which $\alpha$ is received by $w$ again (due to the fact that $a$ transmits every $\texttt{delay}_{t_a}(a)$ rounds, and the argument from the previous paragraph that there will be no interference from $b$, $w$, and $w$'s parent). When $\alpha$ is received, it will not overwrite $\beta$, as it is not a multiple of $\beta$, so it will be appended to \texttt{childrenEncs} instead. So we are now in the position described earlier: $\beta$ appears earlier in the list than $\alpha$, and the same reasoning shows that they will never overwrite each other in any future round.
\end{proof}


Theorem \ref{wsettles} with $w = r$ shows that executing the {\bf Aggregate} subroutine guarantees that the coordinator node eventually possesses the complete set of source messages. However, we have not discussed how the nodes will know when to stop the second stage of the algorithm and begin the third stage (the {\bf Inform} subroutine). There are two challenges: (1) How does the coordinator know when it possesses the complete set of source messages so that it can begin the {\bf Inform} stage? (2) How do we make sure that all of the other nodes have stopped executing {\bf Aggregate} so that they will be listening when the {\bf Inform} stage begins?

To address the first challenge, the labeling scheme simulates the {\bf Aggregate} subroutine described above, and makes note of the last source message $\mu_{last}$ to arrive at the coordinator $r$ (if multiple such source messages arrive within the same transmitted message, then choose one arbitrarily). The labeling scheme uses one bit called $term$, which it sets to 1 at the node $s_{last}$ that started with source message $\mu_{last}$, and sets to 0 at all other nodes. Then, during the actual execution of the algorithm, the node with $term$ bit set to 1 will include with its source message a ``last" message. When $r$ receives a message containing ``last", it knows that it possesses all of the source messages.

To address the second challenge, we introduce a new set of rounds that are interleaved with the rounds of {\bf Aggregate}. In particular, in odd-numbered rounds of the second stage, the {\bf Aggregate} subroutine is executed as described above, and initially, in even-numbered rounds, all nodes listen. Eventually, the coordinator $r$ receives the ``last" message in some round $t$ (which is an odd-numbered round, as this occurs during the execution of {\bf Aggregate}). Then, in round $t+1$, coordinator $r$ transmits a message containing ``finish". Whenever a node receives a ``finish" message for the first time, it stops executing {\bf Aggregate} immediately, and it transmits a ``finish" message two rounds later (in the next even-numbered round). As the height of the tree is $m$, it follows that all nodes have terminated {\bf Aggregate} by round $t+2m-1$ and the last ``finish" message is sent by round $t+2m+1$. Thus, the coordinator can safely begin the {\bf Inform} stage of the algorithm in round $t+2m+2$ (recall that $r$ knows the value of $m$ from the {\bf Initialize} stage).

\begin{theorem}\label{Gossipsolved}
	Consider any $n$-node unlabeled tree $G$, and suppose that each node has an initial source message. By applying the $O(\log D(G))$-bit labeling scheme $\gossipscheme$ and then executing algorithm $\mathcal{GOSSIP}$, all nodes possess the complete set of source messages.
\end{theorem}

%
%

 \bibliographystyle{splncs04}
 \bibliography{main-arxiv}

\end{document}